\documentclass[letterpaper, 10 pt, conference]{ieeeconf}  

\IEEEoverridecommandlockouts                              %
\let\labelindent\relax                                                         

\usepackage{amsthm}
\newif\ifproof
\prooftrue

\usepackage{amsmath,amssymb}
\usepackage{xcolor}
\usepackage{subcaption}
\usepackage{booktabs} \usepackage{tikz}
\usetikzlibrary{automata,positioning}
\usetikzlibrary{shapes,arrows}
\tikzset{state/.style={circle, draw, minimum size=0.5cm, initial distance=0.2cm}}

\usepackage{todonotes}
\usepackage[noadjust]{cite}

\usepackage{algorithm}
\usepackage{enumitem}
\usepackage{algorithmicx}
\usepackage{algpseudocode}
\algnewcommand{\LineComment}[1]{\Statex \(\triangleright\) #1}
\makeatletter
\newcommand{\algmargin}{\the\ALG@thistlm}
\makeatletter
\newcommand*{\rom}[1]{\expandafter\@slowromancap\romannumeral #1@}
\makeatother
\newlength{\whilewidth}
\algdef{SE}[parWHILE]{parWhile}{EndparWhile}[1]
  {\parbox[t]{\dimexpr\linewidth-\algmargin}{     \hangindent\whilewidth\strut\algorithmicwhile\ #1\ \algorithmicdo\strut}}{\algorithmicend\ \algorithmicwhile}\algnewcommand{\parState}[1]{\State  \parbox[t]{\dimexpr\linewidth-\algmargin}{\strut #1\strut}}
\algnewcommand{\parRequire}[1]{\Require  \parbox[t]{\dimexpr\linewidth-\algmargin}{\strut #1\strut}}
\DeclareMathAlphabet{\mathpzc}{OT1}{pzc}{m}{it}
\usepackage[hidelinks]{hyperref}
\usepackage{cleveref}
\usepackage{etoolbox}
\usepackage{stackengine}
\def\delequal{\mathrel{\ensurestackMath{\stackon[1pt]{=}{\scriptstyle\Delta}}}}

\newtheorem{definition}{Definition}
\newtheorem{theorem}{Theorem}
\newtheorem{remark}{Remark}
\newtheorem{assumption}{Assumption}

\newtheorem{lemma}{Lemma}
\newtheorem{proposition}{Proposition}
\newtheorem{corollary}{Corollary}
\newtheorem{example}{Example}

\crefname{section}{Section}{Sections}
\crefname{subsection}{Section}{Sections}
\crefname{definition}{Definition}{Definitions}
\crefname{proposition}{Proposition}{Propositions}
\crefname{lemma}{Lemma}{Lemmas}
\crefname{theorem}{Theorem}{Theorems}
\crefname{corollary}{Corollary}{Corollaries}
\crefname{example}{Example}{Examples}
\crefname{figure}{Figure}{Figures}
\crefname{assumption}{Assumption}{Assumptions}
\crefname{remark}{Remark}{Remarks}
\crefname{running}{Running Example}{Running Examples}
\crefname{algorithm}{Algorithm}{Algorithms}

\newcommand\old[1]{{\color{gray} #1}}

\renewcommand\old[1]{}

\title{\LARGE \bf
Attacks on Perception-Based Control Systems: Modeling and Fundamental Limits
}

\author{Amir Khazraei, Henry Pfister, and Miroslav Pajic
\thanks{This work is sponsored in part by the ONR under the agreement N00014-20-1-2745, AFOSR under the award number FA9550-19-1-0169, by the NSF under the CNS-1652544 award as well as the National AI Institute for Edge Computing Leveraging Next Generation Wireless Networks, Grant CNS-2112562. Some of preliminary results in this
paper appeared in~\cite{khazraei_l4dc22}.}
\thanks{The authors are with the Department of Electrical and Computer
Engineering, Duke University, Durham, NC 27708 USA (e-mail:
amir.khazraei@duke.edu, henry.pfister@duke.edu, miroslav.pajic@duke.edu). }}

\begin{document}

\maketitle
\thispagestyle{empty}
\pagestyle{empty}

\begin{abstract}
We study the performance of perception-based control systems in the presence of attacks, 
and provide methods for modeling and analysis of their 
resiliency to stealthy attacks on both physical and perception-based sensing. 
Specifically, we consider a general setup with 
a nonlinear affine physical plant controlled with a perception-based controller that maps both the physical (e.g., IMUs) and perceptual (e.g., camera) sensing to the control input; the system is also equipped with a statistical or learning-based anomaly detector (AD). We model the attacks 
in the most general form, and introduce the notions of attack effectiveness and stealthiness 
independent of the used~AD.

In such setting, we consider attacks with different levels of runtime knowledge about the plant. 
We find sufficient conditions for existence of stealthy effective attacks that force the plant into an unsafe region without being detected  by \emph{any} 
AD. We show that
as the {open-loop} unstable 
plant dynamics diverges faster and the closed-loop system converges faster to an equilibrium point, the system is more 
vulnerable to  effective stealthy attacks. 
{Also, depending on runtime information available to the attacker, the probability of attack remaining stealthy 
can be arbitrarily close to one, if the attacker’s estimate of the plant's state is arbitrarily close to the true state; when 
an accurate estimate of the plant state is not available, the stealthiness level depends on the control performance in attack-free operation.} 

\end{abstract}

\section{Introduction}
\label{sec:introduction}

The recent progress in deep-learning and computer vision has created a new generation of control systems that incorporate perceptual data for control and decision-making. For example, a camera or a light detection and ranging (LiDAR) sensor 
can provide the controller with information about 
plant states (e.g., distance to the obstacles, position in a local frame). 
Deep neural networks (DNNs) have shown capability to extract information from the complex perception data such as images. 

Two main approaches -- i.e., \emph{modular} and \emph{end-to-end} perception-based controllers have been recently developed~\cite{tampuu2020survey}. With the end-to-end control approach, perception data (e.g., camera images, LiDAR 3D-point clouds), combined with other physical sensor information, is directly mapped to control inputs at runtime (e.g., see~\cite{rausch2017learning,jaritz2018end,codevilla2018end}). 
The controller is then either learned via supervised methods, such as imitation learning that mimics a pre-obtained optimal control input, or using deep-reinforcement learning techniques that design the control policy by maximizing the cumulative reward in an unsupervised fashion. On the other hand, with the modular control approach, a subset of state information is extracted from the perception data (e.g., images) and then combined with other physical sensor information, followed by the use of classic feedback controllers (e.g.,~\cite{dean2020certainty,dean2020robust,al2020accuracy}).

Despite~the tremendous promise that DNN-based perception 
brings to controls, 
resiliency of perception-based controllers to well-documented adversarial threats remains a challenge, limiting their applicability in real-world scenarios. 
The main focus of 
adversarial machine learning methods 
has been on vulnerability of DNNs to small input perturbation, effectively focusing on robustness analysis of DNNs; e.g., targeting DNNs classification or control performance  
when a small and bounded noise is added to the images in camera-based control systems (e.g., ~\cite{boloor2020attacking,jia2020fooling}). 
However, 
an attacker capable of compromising the system's perception/sensing would not limit their actions (i.e., injected data) to small bounded measurement perturbation; the reason is that control theoretic methods for designing stealthy attacks suggests that the perturbation vector should gradually increase in size over time (e.g.,~\cite{mo2010false,jovanov_tac19,khazraei2020perfect,kwon2014analysis}).

Little consideration has been given on the potential impact of stealthy (i.e., undetectable) attacks, which are especially dangerous in the control context as many systems have recovery-based defenses triggered once an attack is detected.  Model-based vulnerability analysis methods, designed from the control theory perspective, have been used 
to analyze the impact stealthy attacks could have on systems with linear time-invariant (LTI) dynamics and simple physical sensing, without perception, (e.g.,~\cite{jovanov2017sporadic,khazraei2022attack,mo2010false,kwon2014analysis}). However, such analysis cannot be easily extended to systems with complex 
dynamics and sensing model that includes 
perception-based sensing.

Consequently, this work 
studies the impact that stealthy~attacks on system sensing could have on perception-based~control. 
We assume the attack goal is to move the system into an unsafe region while remaining stealthy from 
\textbf\textit{any} 
anomaly detector (AD). Even though our notion of stealthiness is related to the work~\cite{bai2017data}, we do not restrict the attack impact 
to the infinite time horizon. 
We show how such attacks can be modeled 
in a general form of an \emph{additive term} for physical (i.e., non-perception) sensors and a \emph{generative model} for perception sensors, without any assumptions about a bound on the perturbation values. 

Perception-based controllers (either end-to-end or modular controllers) implicitly extract the state information from perception sensing (e.g., camera images) to derive suitable control inputs. Hence, the goal of this work is to evaluate their vulnerability to sensing attacks, by investigating whether there exist adversarial perception outputs 
that 
convey desired `falsified' state information while remaining stealthy, rather than blindly adding noise or e.g., a patch to the current image. 
For example, for a vehicle with lane-keeping~control (whose goal is to remain in the lane center), we find a sufficient stealthy attack sequence based on suitable adversarial image sequences conveying the desired falsified state information that 
is stealthy while fooling the controller into unsafe steering~decisions.

Depending on attacker's level of knowledge about the plant's states, 
we derive conditions for which there exists a stealthy and effective (i.e., impactful) attack sequence that forces the system far from the operating point in the safe region. In particular, we assume that the attacker has full knowledge about the system's open-loop dynamics, 
and consider two cases where at runtime the attacker (\emph{i})~has or (\emph{ii})~does not have access to the estimation of the plant's state. We show that 
in the first case, the attack can approach arbitrarily close to the strict stealthiness condition as long as the estimation error is small. 
For the latter, 
the stealthiness level of the attack depends on the system's performance in attack-free operation as well as the 
level of impact that the attacker expects to impose on the system ({in terms of the distance from the system's desired operating point due to the attack}). Thus, there is a trade-off between the stealthiness guarantees and the performance degradation 
caused by~the~attack.

Moreover, 
for LTI plants, 
we show that these two cases fall in the same category, which means the attacker does not need to have access to an accurate estimate of the plant's states. 
We also show that unlike 
in systems with LTI plants and linear controllers, where the design of stealthy and effective attacks is independent of the control design, for nonlinear plants the level of stealthiness is closely related to the level of closed-loop system stability -- i.e., if the closed-loop system `is more stable' (i.e., its trajectory {converges faster to the equilibrium point}), 
the attack can have stronger stealthiness guarantees. On the other hand, the attack impact (i.e., control degradation) 
fully depends on the level of open-loop system instability (e.g., the size of unstable eigenvalues for LTI systems). 


\subsection{Related Work} 
The initial work~\cite{szegedy2013intriguing} on adversarial example generation showed that 
DNNs are vulnerable to small input perturbations. Afterward, the majority of works has applied this idea to adversarial attacks on physical world such as malicious stickers on traffic signs to fool the detectors and/or classifiers~\cite{eykholt2018robust, papernot2017practical,sun2020towards}. %
For example, 
design of adversarial attacks for regression task has been studied in~\cite{wong2020targeted}, where the goal of the attacker is to alter the  geometrical prediction of the scene and the predicted distances from the camera. 
Yet, all these methods only consider classification or regression tasks in a static manner;  
i.e., the 
target only depends on its input,  without  consideration of the longitudinal (i.e., over time) system behaviours. 

The vulnerability of perception-based 
vehicle controls has been recently studied in longitudinal way (e.g., \cite{boloor2020attacking,jia2020fooling,yoon2021learning, hallyburton_security22}). For instance,~\cite{boloor2020attacking} 
considers autonomous vehicles with end-to-end DNN controllers that directly map perceptual inputs into 
the vehicle steering angle, and target the systems 
by painting black lines on the
road. 
On the other hand,~\cite{yoon2021learning,jia2020fooling} introduced online attacks on streams of images in the tracking task, while~\cite{hallyburton_security22} considers longitudinal attacks on camera-LiDAR fusion. Specifically,~\cite{jia2020fooling} uses the idea of adding a small patch into the image for a couple of consecutive frames to change the position of the bounding box around each object, where the location and the size of the patch is obtained by solving an optimization problem.

However, these works only consider specific applications and analyze the attack impact in an ad-hoc manner, limiting the use of their results in other systems/domains. Further, they lack any consideration of attack stealthiness, as injecting e.g., adversarial patches that only maximises the disruptive impact on the control can be detected by most ADs. For instance,~\cite{cai2020real} introduced an AD that easily detects the adversarial attacks from~\cite{boloor2020attacking}. On the other hand, in this work, we focus on systems with nonlinear system dynamics, define general notions of attack stealthiness, and introduce sufficient conditions for a perception-based control system to 
be vulnerable to 
effective yet stealthy perception and sensing attacks. 
We show that to 
launch a successful stealthy attack, the attacker may need to compromise other sensing information besides perception.


Finally, for \emph{non-perception} control systems,  stealthy attacks 
have been well-defined in e.g.,~\cite{mo2009secure,mo2010false,teixeira2012revealing,khazraei2022attack,smith2015covert,bai2017data,sui2020vulnerability,khazraei2021learning,khazraei2020perfect,pajic_tcns17,liu2020secure,bianchin2019secure,gadginmath2022direct}, including replay~\cite{mo2009secure}, covert~\cite{smith2015covert}, zero-dynamic~\cite{teixeira2012revealing}, and false data injection attacks~\cite{mo2010false,jovanov_tac19, khazraei2022attack}. 
However, all these work also only focus on LTI systems and linear controllers, as well as on specific AD design (e.g., $\chi^2$ detector). The problem of resilient state estimation for a class of nonlinear control systems has been considered in e.g.,~\cite{hu2017secure},
focusing on systems that do not include perception as part of the closed-loop.

The notion of attack stealthiness independent of the employed AD (i.e., remaining stealthy for \emph{all} existing/potential ADs) has been studied in~\cite{bai2017data}. 
We additionally differentiate our work in the following -- our notion of stealthiness is stronger than the one in~\cite{bai2017data} as 
stealthiness there 
depends on time;  i.e., there exists only a bounded time that the attack can stay undetected by  
an AD. However, the notion of stealthiness in our work is independent of time and the attack is guaranteed to be stealthy for all time steps after initiating the attack.  Moreover, the performance degradation metric used in~\cite{bai2017data} is the error covariance of a Kalman filter estimator as opposed in our work; we assume the attacker's goal is to force the system states into an unsafe region. 
In addition, we consider systems with nonlinear dynamics 
as well as with both perception and physical sensors, 
unlike 
\cite{bai2017data}
where only LTI plants with physical (i.e., without perception) sensors are~investigated. 

Further, this work considers attacks on both perception (e.g., images) and sensor measurements while 
the existing vulnerability analyses do not take into account perception. Besides, these works all assume LTI model and linear controller for the system dynamics while we consider a class of systems with nonlinear dynamical model and control. In addition, they focus  on  detectability of attacks with respect to specific detection schemes employed by the controller, such as the classic $\chi^2$ 
anomaly detection algorithm. 
Recently~\cite{khazraei2021learning} has introduced a learning-based attack design for systems with nonlinear dynamics; yet, in addition to providing no formal analysis, the work only considers stealthiness with respect to the $\chi^2$-based AD, and does not consider perception-based~controllers.

\subsection{Paper Contribution and Organization} The contribution of this work is twofold. First, we consider attacks on perception-based control in the presence of ADs. We define a new notion of stealthiness using  Neyman-Pearson Lemma and its relation with total variation and Kullback-Leibler (KL) divergence, where an attack is considered stealthy if it is stealthy from \emph{any} 
AD. We find a sufficient condition for which the perception-based control systems are vulnerable to highly effective, in terms of moving the system from the desired operating point, yet stealthy attacks. Second, unlike all previous works in control literature, we consider highly impactful stealthy attacks on plants with \emph{nonlinear} dynamics, controlled by a perception-based controller 
potentially employing an  end-to-end nonlinear controller (e.g., DNN).

The paper is organized as follows. In Section~\ref{sec:problem_des}, we present the system  model, as well as a model of  attacks on perception before  introducing the  concept  of  stealthy yet effective attacks on the control systems. 
Section~\ref{sec:perfect} introduces two attack strategies and provides conditions under which the system is vulnerable to highly effective stealthy attacks. Finally, in Section~\ref{sec:simulation}, we provide case studies to illustrate these conditions, before concluding remarks in Section~\ref{sec:conclusion}.

\section{Preliminaries}

In this section, we introduce employed notation before presenting 
properties of 
KL divergence known as \emph{data processing inequality}, \emph{monotonicity} and \emph{chain-rule} as presented in~\cite{polyanskiy2022information}.

\paragraph*{Notation}
$\mathbb{R}$ denotes the set of reals, whereas $\mathbb{P}$ and $\mathbb{E}$ denote the probability and expectation of a random variable.  For 
a square matrix $A$, $\lambda_{max}(A)$ is the maximum eigenvalue.  
For a vector $x\in{\mathbb{R}^n}$, $||x||_p$ denotes the $p$-norm of $x$; when $p$ is not specified, the 2-norm is implied. 
For a vector sequence, 
$x_0:x_t$ denotes the set $\{x_0,x_1,...,x_t\}$. 
A function $f:\mathbb{R}^{n}\to \mathbb{R}^{p}$ is Lipschitz on the set $\mathcal{D}$ with constant $L$ if for any $x,y\in \mathcal{D} \subseteq \mathbb{R}^{n}$ it holds that $||f(x)-f(y)||\leq L ||x-y||$; it is globally Lipschitz with constant $L$ if $\mathcal{D} = \mathbb{R}^{n}$. If $X$ and $Y$ are two sets, $X-Y$ includes the elements in $X$ that are not in $Y$. For a set $X$, $\partial X$ and $X^o$ define the boundary and the interior of the set, respectively. $B_r$ denotes a closed ball centered at zero with radius $r$; i.e., $B_r=\{x\in \mathbb{R}^n\,\,|\,\,\Vert x\Vert \leq r\}$, whereas $\mathbf{1}_A$ is the indicator function on a set $A$.
For a function $f$, we denote $f'=\frac{\partial f}{\partial x}$ as the partial derivative of $f$ with respect to $x$ and $\nabla f_i(x)$ is the gradient of the function $f_i$ ($i$-th element of the function $f$). 
Finally, if $\mathbf{P}$ and $\mathbf{Q}$ are probability distributions relative to Lebesgue measure with densities $\mathbf{p}$ and $\mathbf{q}$, respectively, then the total variation between them is defined as $TV(\mathbf{P},\mathbf{Q})=\frac{1}{2}\int \vert \mathbf{p}(x)-\mathbf{q}(x)\vert dx$. The 
KL divergence between $\mathbf{P}$ and $\mathbf{Q}$ is defined as
$KL(\mathbf{P}||\mathbf{Q})=\int \mathbf{p}(x)\log{\frac{\mathbf{p}(x)}{\mathbf{q}(x)}}dx$. 

\subsection*{Properties of KL Divergence and Other Preliminaries} 
In the following lemmas, $X$ and $Y$ are assumed to be random variables (Lemmas~\ref{lemma:data}-\ref{lemma:Guassian} proofs can be found in~\cite{polyanskiy2022information}).

\begin{lemma}\label{lemma:data}[\textbf{Data Processing Inequality}]
Assume $Y$ is produced given $X$ based on the conditional law $W_{Y|X}$. Let $P_Y$ (respectively $Q_Y$) denote the distribution of $Y$ when $X$ is distributed as $P_X$ (respectively $Q_X$).
Then
\begin{equation}
KL(Q_Y||P_Y)\leq KL(Q_X||P_X).
\end{equation}
\end{lemma}
\begin{lemma}\label{lemma:mon} [\textbf{Monotonicity}] 
Let $P_{X,Y}$ and $Q_{X,Y}$ be two distributions for a pair of variables $X$ and $Y$, and $P_{X}$ and $Q_{X}$ be the marginal distributions for variable $X$.  Then, 
\begin{equation}
KL(Q_X||P_X)\leq KL(Q_{X,Y}||P_{X,Y}).
\end{equation}
\end{lemma}

\begin{lemma}\label{lemma:chain}[\textbf{Chain rule}] 
Let $P_{X,Y}$ and $Q_{X,Y}$ be two distributions for a pair of variables $X$ and $Y$. Then,
\begin{equation}
KL(Q_{X,Y}||P_{X,Y})= KL(Q_{X}||P_{X})+KL(Q_{Y|X}||P_{Y|X}),
\end{equation}
\end{lemma}
where $KL(Q_{Y|X}||P_{Y|X})$ is defined as
\begin{equation}
KL(Q_{Y|X}||P_{Y|X})=\mathbb{E}_{x\sim Q_X}KL(Q_{Y|X=x}||P_{Y|X=x}).
\end{equation}

\begin{lemma}\label{lemma:Guassian}
Let $P_{X}$ and $Q_{X}$ be two Gaussian distributions with the same covariance $\Sigma$ and different means $\mu_Q$ and $\mu_P$, respectively. Then, it holds that
\begin{equation}
KL(Q_{X}||P_{X})= \mu_Q^T\Sigma^{-1} \mu_P.
\end{equation}
\end{lemma}

\begin{lemma}\label{lemma:maximum}
Let $Q_{X}$ be a distribution for a scalar random variable $X$, and that $X\leq M$ for some $M>0$. Then, 
\begin{equation}
\mathbb{E}_{Q_X}\{X\}\leq M.
\end{equation}
\end{lemma}
\begin{proof}
The proof is straightforward 
from the definition of expectation and some properties of the integral.  
\end{proof}


\section{Modeling Perception-Based Control Systems in the Presence of Attacks}\label{sec:problem_des}

In this section, we introduce the system model and show how to capture attacks on system sensing, including perception. 
Specifically, we consider the setup from Fig.~\ref{fig:architecture} where each of the components is modeled as follows.

\subsection{Plant and Perception Model}

We assume the plant 
has nonlinear dynamics in the standard state-space form
\begin{equation}
\label{eq:plant}
\begin{split}
x_{t+1} &= f(x_{t})+Bu_{t}+w_{t}, \quad
y_{t}^s = C_sx_{t}+v_t^s,\\
z_{t} &= G(x_{t}).
\end{split}
\end{equation}
Here, $x_t \in {\mathbb{R}^n}$, $u_t \in {\mathbb{R}^m}$, $w_t\in {\mathbb{R}^n}$, $z_t \in {\mathbb{R}^l}$, $y_t^s\in {\mathbb{R}^s}$ and $v^s\in {\mathbb{R}^s}$ denote the state, input, system disturbance, 
observations from perception-based sensors, (non-perception) sensor measurements, 
and sensor noise, at time $t$, respectively. The 
perception-based sensing is 
modeled by
an unknown generative model $G$, which is nonlinear and potentially 
high-dimensional. For example, consider a camera-based 
lane keeping system. 
Here, the observations $z_{t}$ are the captured images; the map $G$ generates the images based on the vehicle's position. Without loss of generality, we assume that $f(0)=0$. Finally, 
the process and measurement noise vectors $w$ and $v^s$ are assumed independent and identically distributed (iid)  Gaussian processes $w\sim \mathcal{N}(0,\,\Sigma_w)$ and $v^s\sim \mathcal{N}(0,\,\Sigma_{v^s})$.





\subsection{Control Unit}

The control unit, shown in Fig.~\ref{fig:architecture}, consists of perception, controller, and anomaly detector units. We assume that the control unit receives $y^{c,s}$ and $z^{c}$ as the sensor measurements and perception sensing (e.g., images), respectively as an input. 
Thus, without malicious activity, it holds that $y^{c,s}=y^s$ and $z^c=z$. Now, we describe in detail each of the components.


\subsubsection{Perception}
We assume that there exists a perception map $P$ that imperfectly estimates the partial state information from from perception sensing (e.g., images) --
i.e.,
\begin{equation}
\label{eq:perception}
y_t^P=P(z_t^c) = C_Px_t + v^P(x_t);
\end{equation}
here, $P$ denotes 
a deep neural network (DNN) 
trained using any supervised learning method on a data set $\mathcal{X}=\{(z_i,x_i)\}_{i=1}^{N}$  collected densely around the operating point $x_o$ of the system, as in~\cite{dijk2019neural,lambert2018deep}. In addition, $v^P \in \mathbb{R}^{p}$ is the perception map error that  depends on the state of the system -- i.e., smaller around the training data set.
To capture perception guarantees, 
we employ the model for robust perception-based control from~\cite{dean2020certainty}, 
and the standard model of the perception map~\eqref{eq:perception} from~\cite{dean2020robust} 
that can capture well DNN-based perception modules. 
Specifically, if the model is trained effectively, we assume that the perception error $v^P$ around the operating point $x_o$ is bounded, 
i.e., the following assumption from~\cite{dean2020robust} holds. 

\begin{assumption}
\label{ass:perception}
There exists a safe set $\mathcal{S}$ with the radius  $R_{\mathcal{S}}$ (i.e., $\mathcal{S}=\{x\in \mathbb{R}^n\,\,|\,\,\Vert x\Vert \leq R_{\mathcal{S}}\}$) around the operating point such that for all $x \in \mathcal{S}$, it holds that $\Vert P(z)-C_Px\Vert\leq \gamma$, where $z=G(x)$ -- i.e., for all $x\in\mathcal{S}$, $\|v^P(x)\| \leq \gamma$. 
Without loss of generality, in this work, we consider the origin as the operating point -- i.e.,  $x_o=0$.
\end{assumption}

\begin{figure}[!t]
\centerline{\includegraphics[width=\columnwidth]{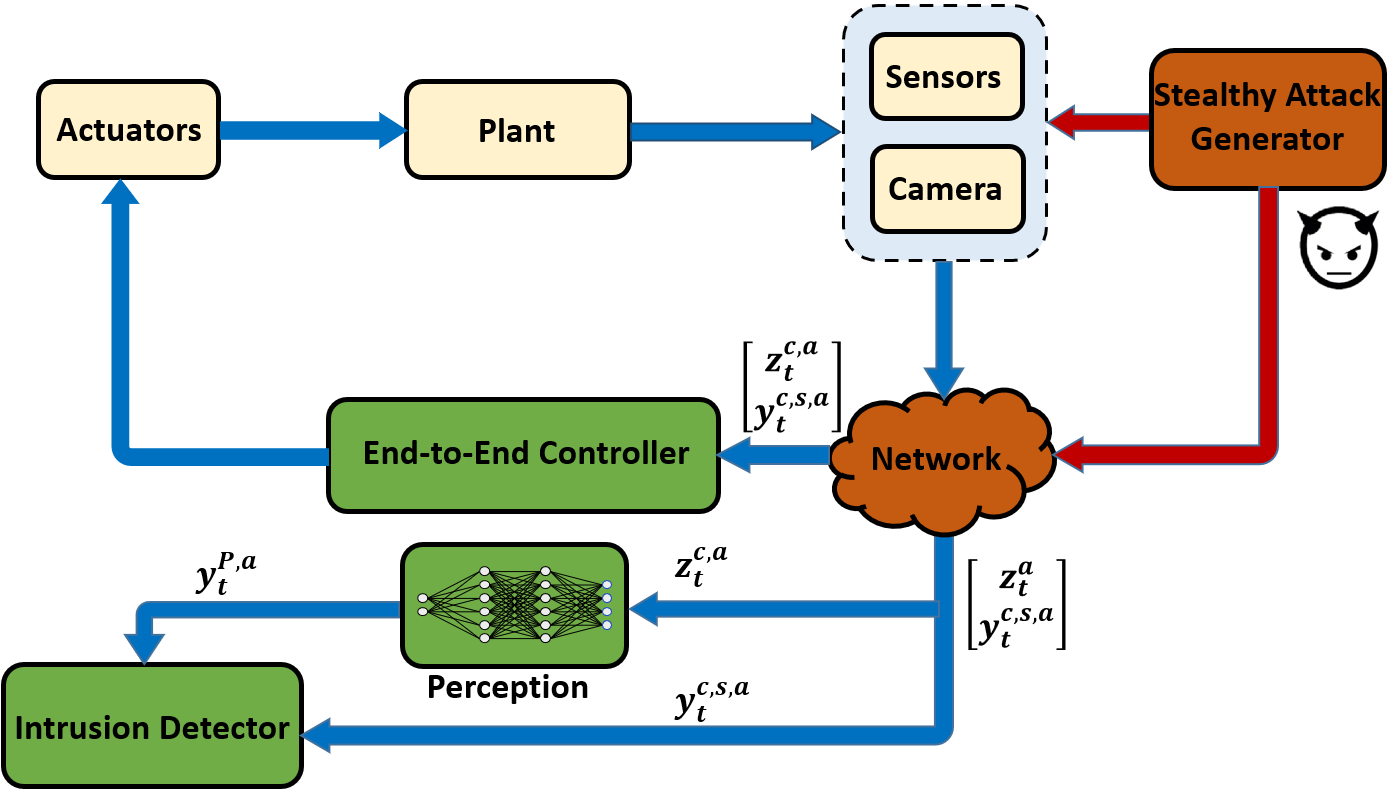}}
\caption{The architecture of a perception-based control system under attack on system sensing, including perception. Independently of the way attacks are actually implemented (i.e., directly compromising a sensor or modifying the measurements delivered over the network to the controller), the same impact on the control performance is obtained.} 
\label{fig:architecture}
\end{figure}

\begin{remark}
As will be shown later in this work, the assumption on boundedness of $\|v^P(x)\|$  and having a good estimate of such bound is only essential for the defender because the perception map is also used for anomaly detection. We should note that such bound will not be employed by the attacker to design stealthy impactful attacks. A systematic method to find the bound on $\|v^P(x)\|$ is discussed in~\cite{dean2020robust}.
\end{remark}
\subsubsection{Controller} 
The physical plant~\eqref{eq:plant} is controlled by a (general) nonlinear controller $u_t=\pi(z_t^c,y_t^{c,s})$ that maps the image and physical sensor information to the input control.  Using~\eqref{eq:plant}, one can write the above control law in an equivalent form of $u_t=\Pi(x_t,v_t^s)$ that absorbs $z_t^c=G(x_t)$ in the control function $\pi$. Hence, 
for 
$$h(x_t,v_t^s)=f(x_t)+B\pi(z_{t}^c,y_t^{c,s})=f(x_t)+B\Pi(x_{t},v_t^s),$$
the evolution of the closed-loop system can be captured~as 
\begin{equation}\label{eq:controller}
\begin{split}
x_{t+1}&=h(x_t,v_t^s)+w_t.
\end{split}
\end{equation}

In the general form, the controller can 
employ \textbf{\textit{any}} end-end 
control policy that uses the image and sensor measurements. For noiseless systems,  
the state dynamics 
can be captured~as\footnote{With slight abuse of notation, $x$ is used to denote the states of the noiseless system; yet, in the rest of the paper we use $x$ to denote the states of the actual physical system with noise and we clarify if the other case is implied.} 
\begin{equation}\label{eq:noiseless}
x_{t+1}=h(x_t,0).
\end{equation}


\begin{definition}
The origin of the system~\eqref{eq:noiseless} is 
exponentially stable on a set $\mathcal{D}\subseteq \mathbb{R}^n$ if for any $x_0\in \mathcal{D}$, there exist $0< \alpha<1$ and $M>0$, such that $\Vert x_t \Vert \leq M\alpha^t\Vert x_0\Vert $ for all $t\geq 0$.
\end{definition}

\begin{lemma}[\cite{kushner2014partial}]
\label{lemma:aymp}
For the system from~\eqref{eq:noiseless}, if there exists a function $V: \mathbb{R}^n\to \mathbb{R}$ such that for any $x_t\in \mathcal{D}\subseteq \mathbb{R}^n$, the following holds 
\begin{equation}\label{eq:EXP_stable}
\begin{split}
c_1\Vert x_t\Vert^2\leq V(x_t)&\leq c_2\Vert x_t\Vert^2, \\  V(x_{t+1})-V(x_t)&\leq -c_3\Vert x_t\Vert^2,\\ 
\Vert \frac{\partial V(x)}{\partial x} \Vert &\leq c_4\Vert x\Vert,
\end{split}
\end{equation}
for some positive $c_1, c_2$, $c_3$ and $c_4$, then the origin is exponentially stable.
\end{lemma}

\begin{assumption}\label{ass:control}
We assume that for the closed-loop control system~\eqref{eq:noiseless} the origin is exponentially stable on a set $\mathcal{D}=B_d$.  Using the converse Lyapunov theorem~\cite{khalil2002nonlinear}, there exists a Lyapunov function that satisfies the inequalities in~\eqref{eq:EXP_stable} with constants $c_1$, $c_2$, $c_3$, and $c_4$ on a set $\mathcal{D}=B_d$. As a result, it follows that $\Pi(0,0)=0$.

\end{assumption}

\begin{remark}
The assumptions 
for closed-loop system are critical for 
system guarantees without the attack; i.e., if the system does not satisfy the stability property in attack-free scenarios, then 
an effective strategy for the attacker would be to wait until the system fails by itself. We refer the 
reader to the recent work e.g.,~\cite{dean2020certainty, dean2020robust} 
on design of such controllers.
\end{remark}

\begin{remark}
Note that the exponential stability assumption on the closed-loop system~\eqref{eq:noiseless} is only considered due to the simplicity of the notation and can be relaxed to control systems with asymptotic stability  that satisfies the converse Lyapunov theorem conditions~\cite{khalil2002nonlinear} (Theorem 3.14). {Specifically, a similar result as in Lemma~\ref{lemma:khalil} can be obtained; however, the needed notation would be significantly more~cumbersome.} 
\end{remark}





\begin{definition}\label{def:unstable_fun}
Let $\mathcal{U}_{\rho}$ be the set of all functions $f:\mathbb{R}^n\to \mathbb{R}^n$ such that the dynamics $x_{t+1}=f(x_t)+d_t$, where $d_t$ satisfies $\Vert d_t\Vert \leq \rho$, 
reaches arbitrarily large states for some nonzero initial state $x_0$. For 
a function $f$ from $\mathcal{U}_{\rho}$ 
and an initial condition $x_0$,
we define 
\begin{equation}
\label{eq:Tf}
T_f(\alpha,x_0)=\min\{t\,| \,\Vert x_t\Vert\geq \alpha \}. 
\end{equation}
i.e., $T_f(\alpha,x_0)$ 
is the minimal 
number of time-steps needed for an unstable dynamic $f$, starting from the initial condition $x_0$, to leave a bounded ball  with the center at zero and radius~$\alpha$.\footnote{To simplify our notation, and since we consider specific $f$ from the plant dynamics~\eqref{eq:plant}, we drop the the subscript $f$.} 
\end{definition}



\subsubsection{Anomaly Detector}
The system is equipped with an anomaly detector (AD) 
designed to detect the presence of any abnormal behaviours. 
We use  $Y_t=\left[\begin{smallmatrix}y_t^P\\y_t^{c,s}\end{smallmatrix}\right]$ and  $Y_t^a=\left[\begin{smallmatrix}y_t^{P,a}\\y_t^{c,s,a}\end{smallmatrix}\right]$ to capture sensor (from~\eqref{eq:plant}) and perception-based (from~\eqref{eq:perception}) values without and under attack, respectively --  we use the superscript $a$ to differentiate all signals of the attacked system, with the full attack model introduced in the next subsection.
Now, by denoting $Y_{-\infty}^{-1}=Y_{-\infty}:Y_{-1}$, we consider the 
classical binary hypothesis testing problem:\\

$H_0$:  normal condition (the AD receives $Y_{-\infty}:Y_t$);~~

$H_1$: abnormal behaviour (the AD receives 
$Y_{-\infty}^{-1},Y_{0}^a:Y_t^a$).
\\

\noindent Effectively, the AD uses both the extracted state 
information from the perception map (i.e., \eqref{eq:perception}) as well as sensor measurements. Given a random sequence $\bar{Y}^t = (\bar{Y}_{-\infty}:\bar{Y}_t)$, 
it either comes from the distribution $\mathbf{P}$ (null hypothesis $H_0$), {which is determined by system uncertainties,} or from a distribution $\mathbf{Q}$ (the alternative hypothesis $H_1$); {note that the \emph{{unknown}} distribution $\mathbf{Q}$ is controlled by the attacker}.

For a given AD specified by a function $D: \bar{Y}^t\to \{0,1\}$, 
two types of error may occur. Error type ($\rom{1}$), also referred as \emph{false alarm}, occurs if $D(\bar{Y}^t)=1$ when $\bar{Y}^t \sim \mathbf{P}$;  whereas type ($\rom{2}$) error (\emph{miss detection}) occurs if $D(\bar{Y}^t)=0$ when $\bar{Y}^t \sim \mathbf{Q}$. Hence, 
the sum of the \emph{conditional error probabilities} of AD $D$ for a given random sequence $\bar{Y}^t$ 
is
\begin{equation}\label{eq:pe}
p_t^e(D)=\mathbb{P}(D(\bar{Y}^t)=0|\bar{Y}^t \sim \mathbf{Q})+\mathbb{P}(D(\bar{Y}^t)=1|\bar{Y}^t \sim \mathbf{P}).
\end{equation}

Note that $p_t^e(D)$ is not a probability measure as it can take values larger than one. Let us define $p_t^{TD}(D)=\mathbb{P}(D(\bar{Y}^t)=1|\bar{Y}^t \sim \mathbf{Q})$ as the probability of true detection, and $p_t^{FA}(D)=\mathbb{P}(D(\bar{Y}^t)=1|\bar{Y}^t \sim \mathbf{P})$ as the probability of false alarm for the detector $D$. We say that an AD (defined by $D$) to be better than a random guess-based AD (defined by $D_{RG}$) if $p^{FA}(D)< p^{TD}(D)$; as with the random guess it holds that 
\begin{equation*}
\begin{split}
p^{FA}(D_{RG}) &= \mathbb{P}(D_{RG}(\bar{Y}^t)=1|\bar{Y}^t \sim \mathbf{P}) =\mathbb{P}(D_{RG}(\bar{Y}^t)=1)\\
&= \mathbb{P}(D_{RG}(\bar{Y}^t)=1|\bar{Y}^t \sim \mathbf{Q})=p^{TD}(D_{RG}). 
\end{split} 
\end{equation*}




\vspace{-12pt}
\subsection{Attack Modeling}
We assume that the attacker has the ability to compromise perception-based sensing (e.g., camera images) as well as (potentially) the sensor measurements $y^s_t$ delivered to the controller (see Fig.~\ref{fig:architecture}). 
Such attacks on physical sensors can be achieved by directly compromising the sensing (or the environment of the sensors) or the communication between the sensors and the controller~\cite{lesi_rtss17,lesi_tcps19}. On the other hand, implementing the attack on the images delivered to the controller may not be feasible using physical spoofing attacks such as LiDAR spoofing by injecting laser data points~\cite{sun2020towards,hallyburton_security22}. Hence, the attacker needs to compromise the images in cyber domain (e.g.,~as discussed in~\cite{hallyburton2023partialinformation}). This can be achieved by modifying the firmware directly on the camera or the driver code, as done in Hyundai's Gen5W and Tesla's Model~3 attacks where custom firmware was installed~\cite{CVE1,CVE2}, or the Tesla Model~S attacks that sent custom messages by compromising application drivers~\cite{CVE3}.
Similarly, such cyber attacks can be achieved using \emph{Man-in-the-Middle} attacks that target the connection between the camera and the controller; examples include intercepting legitimate messages and manipulating their contents, before resending to the intended destination (e.g., as described in inter-vehicle attack scenarios in~\cite{lesi_rtss17,lesi_tcps19} or vehicle attacks that manipulated timing on V2V/V2I messages, causing out-of-date situational awareness~\cite{sumra2011timing}.

Moreover, we assume the attack starts at $t=0$, and as previously mentioned, we use the superscript $a$ to differentiate all signals of the attacked system, for all $t\geq 0$; 
the attack sequence is $\{z_{t}^{c,a},y_{t}^{c,s,a}\}_{t\geq 0}$, 
where e.g., 
the value of observation delivered to the perception unit at time $t$ 
is denoted by $z_t^{c,a}$. 
Note that due to nonlinearity of the  operators ($P$ and $G$), we do not employ the additive false-data injection model for perception attacks, widely used for LTI systems with non-perception sensing (e.g.,~\cite{mo2010false}).

Thus, 
the system dynamics under attack can be modeled~as
\begin{equation}
\label{eq:plant_attack}
\begin{split}
x_{t+1}^a &= f(x_{t}^a)+Bu_{t}^a+w_{t}^a,\\
u_{t}^a &= \pi(z_{t}^{c,a},{y_t^{c,s,a}}).
\end{split}
\end{equation}
In this work, we assume the attacker has full knowledge of the system, its dynamics, and employed architecture. Further, the attacker has the required computation power to calculate suitable attack signals to inject,
planning ahead as needed. 

\begin{remark}
The assumption that the attacker has knowledge of the system dynamics (i.e., function $f$ and matrix $C_s$) is commonly used in existing studies focused on the worst-case analysis of the attack impact on control systems~\cite{mo2010false,jovanov_tac19,khazraei2022attack,kwon2014analysis}. In particular, by focusing on resourceful attackers who possess extensive knowledge of the targeted systems, we can effectively assess the consequences of worst-case attacks and develop appropriate defensive strategies. However, the attacks presented in Section~\ref{sec:perfect}
do not require knowledge of the noise profile, significantly improving their applicability.\end{remark}

\begin{remark}
In our 
notation, $x_0^a:x_t^a$ denotes a state trajectory of the system under attack  (for an attack starting at $t=0$), while $x_0:x_t$ denotes the state trajectory of the 
attack-free system; 
we 
refer to such state trajectory 
as the \emph{attack-free} trajectory. 
Thus, when comparing the attack-free trajectory and the system trajectory under attack (i.e., from~\eqref{eq:plant_attack}), 
we 
assume that $w_t^a=w_t$ and $v_t^{s,a}=v_t^{s}$. However, we cannot make such assumption for $v^P(x)$ as it is a function of states and the states are compromised due to the attack.  
\end{remark}

We define an attack to be stealthy if the best strategy for the AD is to ignore the measurements and make a random guess between the hypotheses; i.e., that there is \emph{\textbf{no}} AD $D$ that satisfies $p^{TD}(D)>p^{FA}(D)$. However, reaching such stealthiness guarantees may not be possible in general. Therefore, we define the notion of $\epsilon$-\emph{stealthiness}, which as we will show later, is attainable for a large class of nonlinear systems. Formally, we define the notions of \emph{strict stealthiness} and \emph{$\epsilon$-stealthiness} as~follows.

\begin{definition}
\label{def:stealthy}
Consider the system defined in~\eqref{eq:plant}. An attack sequence is 
\emph{\textbf{strictly stealthy}} if there exists no detector for which $p_t^{FA}<p_t^{TD}$ holds, for any $t\geq 0$. 
An attack is
\textbf{$\epsilon$-\emph{stealthy}} if for a given $\epsilon >0$, there exists no detector such that $p_t^{FA}<p_t^{TD}-\epsilon$ holds, for any $t\geq 0$. 
\end{definition}

Before introducing the sufficient condition for the above notion of stealthiness, we consider the following lemma. 

\begin{lemma}\label{lemma:stealthy}
The anomaly detector $D$ satisfies $ p^{FA}(D)<p^{TD}(D)-\epsilon$ if and only if $p^e(D) < 1-\epsilon$. 
Also, $p^e(D)=1$ if and only if $D$ performs as a random-guess detector. 
\end{lemma}

\begin{proof}
First, we consider the case $p^e(D)< 1-\epsilon$. 
From~\eqref{eq:pe}, it holds that
\begin{equation}
\begin{split}
 &p^e(D)=\mathbb{P}(D(\bar{Y})=0|\bar{Y}\sim \mathbf{Q})+\mathbb{P}(D(\bar{Y})=1|\bar{Y} \sim \mathbf{P})\\
&=1-\mathbb{P}(D(\bar{Y})=1|\bar{Y}\sim \mathbf{Q})+\mathbb{P}(D(\bar{Y})=1|\bar{Y} \sim \mathbf{P})<1-\epsilon\\
\end{split}
\end{equation}
Thus, $ \mathbb{P}(D(\bar{Y})=1|\bar{Y} \sim \mathbf{P})<\mathbb{P}(D(\bar{Y})=1|\bar{Y}\sim \mathbf{Q})-\epsilon$ or $ p^{FA}(D)<p^{TD}(D)-\epsilon$. 

Now, if we have $ p^e(D)=1$, then we get $\mathbb{P}(D(\bar{Y})=1|\bar{Y} \sim \mathbf{P})=\mathbb{P}(D(\bar{Y})=1|\bar{Y}\sim \mathbf{Q})$ where the decision of the detector $D$ is independent of the distribution of $\bar{Y}$ and therefore, the detector performs as the random guess detector. Since the reverse of all these implications hold, 
the other (i.e., necessary) conditions of the theorem also hold.
\end{proof}



Now, we can capture stealthiness conditions in terms of KL divergence of the corresponding distributions.

\begin{theorem}
An attack sequence 
is strictly stealthy if and only if  
    $KL\big(\mathbf{Q}(Y_{-\infty}^{-1},Y_0^a:Y_t^a)||\mathbf{P}(Y_{-\infty}:Y_t)\big)=0$ for all $t\geq 0$.
An attack sequence is $\epsilon$-stealthy if the corresponding observation sequence $Y_{0}^a:Y_t^a$ satisfies
    \begin{equation}\label{ineq:stealthiness}
    \begin{split}
        KL\big(\mathbf{Q}(Y_{-\infty}^{-1},Y_0^a:Y_t^a)||\mathbf{P}(Y_{-\infty}:Y_t)\big)  \leq\log(\frac{1}{1-\epsilon^2}).    
    \end{split}
    \end{equation}
\end{theorem}

\vspace{6pt}
\begin{proof}
With some abuse of notation only specific to this theorem, $\mathbf{Q}$ and $\mathbf{P}$ are used to denote $\mathbf{Q}(Y_{-\infty}^{-1},Y_0^a:Y_t^a)$ and $\mathbf{P}(Y_{-\infty}:Y_t)$, respectively.  First we prove the strictly stealthy case. 

Using Neyman-Pearson lemma for any existing detector $D$, it follows that 
\begin{equation}
p_t^e(D) \geq \int \min\{\mathbf{Q}(y),\mathbf{P}(y)\}dy, 
\end{equation}
where the equality holds for the Likelihood Ratio function as $D^*=\mathbf{1}_{\mathbf{Q}\geq \mathbf{P}}$~\cite{krishnamurthy2017lecture}. 
Since $1-\int \min\{\mathbf{Q}(y),\mathbf{P}(y)\}dy =\frac{1}{2}\int \vert \mathbf{q}(x)-\mathbf{p}(x)\vert dx$, from~\cite{polyanskiy2014lecture} and the definition of total variation distance between $\mathbf{Q}$ and $\mathbf{P}$, it holds that
\begin{equation}\label{ineq:tv}
p_t^e(D)\geq 1-TV(\mathbf{Q},\mathbf{P}),
\end{equation}
where $TV$ denotes the total variation distance between the~distributions. 

Now, it holds that $TV(\mathbf{Q},\mathbf{P})\leq \sqrt{1-e^{-KL(\mathbf{Q}||\mathbf{P})}}$ (Eq. (14.11) in~\cite{lattimore2020bandit}). 
Thus, if $KL\big(\mathbf{Q}||\mathbf{P}\big)=0$ holds, 
then $p_t^e(D)\geq 1$ for any 
detector~$D$. Therefore, according to Lemma~\ref{lemma:stealthy} the attack is strictly stealthy. On the other hand, if for any detector $D$ $p_t^e(D)\geq 1$ holds, then the equality holds for $TV(\mathbf{Q},\mathbf{P})=0$; this is equivalent to $\mathbf{Q}=\mathbf{P}$ and therefore, $KL\big(\mathbf{Q}||\mathbf{P}\big)=0$. 

For the  $\epsilon$-stealthy case, we combine~\eqref{ineq:tv} with the inequality $TV(\mathbf{Q},\mathbf{P})\leq \sqrt{1-e^{-KL(\mathbf{Q}||\mathbf{P})}}$ and the $\epsilon$-stealthy condition~\eqref{ineq:stealthiness}, to show 
$$p_t^e(D)\geq 1-TV(\mathbf{Q},\mathbf{P})\geq 1-\sqrt{1-e^{-KL(\mathbf{Q}||\mathbf{P})}}\geq 1-\epsilon,
$$
therefore, according to Lemma~\ref{lemma:stealthy} the attack is $\epsilon$-stealthy and this concludes the proof.
\end{proof}

\begin{remark}
The $\epsilon$-stealthiness condition defined in~\cite{bai2017data} requires that $$\lim_{t\to \infty}\frac{KL\big(\mathbf{Q}(Y_{0}^a:Y_t^a)||\mathbf{P}(Y_{0}:Y_t)\big)}{t}\leq \epsilon.$$ 
This allows for the KL divergence to linearly increase over time for any $\epsilon>0$, and as a result, after large-enough time period the attack may be detected even though it satisfies the definition of stealthiness from~\cite{bai2017data}. Yet, the $\epsilon$-stealthiness 
from Definition~\ref{def:stealthy} only depends on $\epsilon$ and is fixed for any time $t$; thus, it introduces a stronger notion of stealthiness for the attack. 
\end{remark}

\subsubsection*{Attack Goal} 

We capture the attacker's goal as 
\emph{maximizing} degradation of the control performance. 
Specifically, as 
we consider the origin as the desired operating point of the closed-loop system, 
the attack objective is to maximize the (norm of) states $x_t$. 
Moreover, the attacker wants \emph{to remain stealthy -- i.e., undetected by \textbf{any} employed AD},
as formalized below.

\begin{definition}
\label{def:eps_alpha}
The attack sequence, 
denoted by $\{z_{0}^{c,a},y_{0}^{c,s,a}\}, \{z_{1}^{c,a},y_{1}^{c,s,a}\},...$ is referred to  $(\epsilon,\alpha)$\emph{-successful attack} if there exists $t'\geq 0$ such that $ \Vert x_{t'} \Vert \geq \alpha$  and the 
attack is $\epsilon$-stealthy 
for all $t\geq 0$. 
When such a sequence exists for a system, the system is called $(\epsilon,\alpha)$-\emph{attackable}. 
%
Finally, when the system 
is $(\epsilon,\alpha)$-attackable for arbitrarily large $\alpha$, the system is referred to as \emph{perfectly attackable}. 
\end{definition}

In the rest of this work, our goal is to derive methods to capture the impact of stealthy attacks; specifically, in the next section we derive conditions for 
existence of a \emph{stealthy} yet \emph{effective} attack sequence 
$\{z_{0}^{c,a},y_{0}^{c,s,a}\},\{z_{1}^{c,a},y_{1}^{c,s,a}\},...$  resulting in 
$\Vert x_t \Vert \geq \alpha$ for some $t\geq 0$ -- i.e., we find conditions for a system to be $(\epsilon,\alpha)$-attackable. 
Here, for an attack to be stealthy, we focus on 
the $\epsilon-$stealthy notion; 
i.e., that the best 
AD could only improve the probability detection by $\epsilon$ compared to random-guess baseline detector. 


\section{Conditions for $(\epsilon,\alpha)$-Attackable Systems}
\label{sec:perfect}

To provide sufficient conditions for a system to be $(\epsilon,\alpha)$-attackable, in this section, we introduce 
two  methodologies to design attack sequences on perception and (classical) sensing data. 
The difference in these strategies is the level of information that the attacker has about the system; we show that the stronger attack impact can be achieved with the attacker having full knowledge of the system state.

Specifically, we start with the 
attack strategy where the attacker 
has access to the current estimation of state; 
in such case, we show that the stealthiness condition 
is less restrictive, simplifying design of 
$\epsilon$-stealthy attacks. For the second attack strategy, we show that the attacker can launch the attack sequence with only knowing the function $f$ (i.e., plant model); however, achieving $\epsilon$-stealthy attack in this case is harder 
as more restrictive conditions are imposed on the attacker. 

\subsection{Attack Strategy~${\textit{\rom{1}}}$: Using Estimate of the Plant State}

Consider the attack sequence where
$z^{c,a}_t$ and $y_{t}^{c,s,a}$ injected at time $t$, for all $t\geq 0$,  
satisfy 
\begin{equation}
z^{c,a}_t=G(x_t^a-s_t), \,\,\, y_{t}^{c,s,a}=C_s(x_t^a-s_t)+v_t^{s},
\end{equation}
with $s_{t+1}=f(\hat{x}_t^a)-f(\hat{x}_t^a-s_t)$, and for a nonzero $s_{0}$. Here, $\hat{x}_t^a$ denotes an estimation of the plant's state (in the presence of attacks), and thus ${\zeta}_t=\hat{x}_t^a-x_t^a$ is the corresponding state estimation error. 
Note that the attacker can obtain $\hat{x}_t^a$ 
by e.g., running a local estimator using the true measurements 
before injecting the false values; i.e.,  $y_t^{s,a}=C_sx_t^a+v_t^{s}$ and $z_t^a=G(x_t^a)$. 
We assume that the estimation error is bounded by $b_{\zeta}$ --  i.e., $\Vert \zeta_t\Vert \leq b_\zeta$, for all $t\geq 0$. 
%

On the other hand, 
the above attack design may not require access to the true plant state $x_t^a$, since only the `shifted' (i.e., $x_t^a-s_t$) outputs of the real sensing/perception are  injected. 
For instance, in the lane centring control (i.e., keeping the vehicle between the lanes), $G(x_t-s_t)$ only shifts the actual image $s_t$ to the right or left depending on the coordinate definition. Similarly, the attack on physical (i.e., non-perception) sensors can be implemented as $y_{t}^{c,s,a}=C_sx_t^a+v_t^{s,a}-C_s s_t=y_t^{s,a}-C_s s_t$ where the attacker only needs to subtract $Cs_t$ from the current true measurements. 

The idea behind the above attacks is to have the system believe that its (plant) state is equal to the state $e_t\delequal x_t^a-s_t$; thus, referred to as the \emph{fake state}. Note that effectively both $z_t^a$ and $y_t^{s,a}$ used by an AD are directly functions of the fake state~$e_t$. Thus, if the distribution of $e_0:e_t$ is close to $x_0:x_t$ (i.e., attack-free trajectory), then the attacker will be successful in injecting a  stealthy attack sequence. 


\begin{definition}\label{def:b_x}
For an attack-free state trajectory $x_0:x_t$, and for any $T\geq 0$ and $b_x,b_v>0$,  $\delta(T,b_x,b_v)$ is the probability that the system state and physical sensor noise $v^{s}$ remain in the zero center ball with radius $b_x$ and $b_v$, respectively, during time period $0\leq t\leq T$,  i.e.,  
\begin{equation}
\delta(T,b_x,b_v)\delequal \mathbb{P}\big(\sup_{0\leq t\leq T}\Vert x_t \Vert\leq b_x, \sup_{0\leq t\leq T}\Vert v_t^s \Vert\leq b_v \big). 
\end{equation}
\end{definition}

When the system with exponentially stable closed-loop control dynamics\footnote{Note that under similar conditions the same analysis also holds for asymptotic stability condition~\cite{khalil2002nonlinear}} is affected by a bounded perturbation, one can show that the state of the system will remain in a bounded set. The following lemma from~\cite{khalil2002nonlinear} provides the condition and the upper bound on the norm of the state.

\begin{lemma}[\cite{khalil2002nonlinear}] 
\label{lemma:khalil}
Let $x = 0$ be an exponentially stable equilibrium point of the nominal system~\eqref{eq:noiseless}. Also, let $V(x_t)$ be a Lyapunov function of the nominal system that satisfies~\eqref{eq:EXP_stable} in $\mathcal{D}$, where $\mathcal{D}=B_d$. Suppose the system is affected by additive perturbation term $g(x_t)$ that satisfies $\Vert g(x_t)\Vert \leq \delta+\gamma \Vert x_t\Vert$. 
If $c_3-\gamma c_4>0$ with $\delta<\frac{c_3-\gamma c_4}{c_4}\sqrt{\frac{c_1}{c_2}}\theta d$ holds for all $x\in \mathcal{D}$ and some positive $\theta<1$,  then for all $\Vert x_{0}\Vert <\sqrt{\frac{c_1}{c_2}}d$, there exists $t_1> 0$ such that for all $0\leq t \leq t_1$ $\Vert x_t\Vert \leq \kappa e^{-\beta t} \Vert x_0\Vert $ holds with $\kappa=\sqrt{\frac{c_2}{c_1}}$, $\beta=\frac{(1-\theta) (c_3-\gamma c_4)}{2c_2}$ and for $t\geq t_1$ it holds that  $\Vert x_{t}\Vert\leq b$ with $b=\frac{c_4}{c_3-\gamma c_4}\sqrt{\frac{c_2}{c_1}}\frac{\delta}{\theta}$.
\end{lemma}

The next result captures conditions under which a perception-based control system is not resilient to attacks, in sense that it is ($\epsilon$,$\alpha$)-attackable.

\begin{theorem} \label{thm_closeloop}
Consider the system~\eqref{eq:plant} with closed-loop control as in Assumption~\ref{ass:control}. Assume that the functions $f$, $f'$ (derivative of $f$) and $\Pi'$ (derivative of $\Pi$) are Lipschitz, with constants $L_f$, $L'_f$ and $L'_{\Pi}$, respectively, and let us define $L_{1}=L'_f(b_x+2b_{\zeta}+\phi)$, $L_{2}=\min\{2L_f,L'_f(\alpha+b_x+b_{\zeta}) \}$ and $L_{3}=L'_{\Pi}(b_x+\phi+b_v)$ for some $\phi>0$. Moreover, assume that there exists $\phi>0$ such that the inequalities $L_{1}+L_{3}\Vert B\Vert < \frac{c_3}{c_4}$ and $L_{2}b_{\zeta}<\frac{c_3-(L_{1}+L_{3}\Vert B\Vert)c_4}{c_4}\sqrt{\frac{c_1}{c_2}}\theta d$, for some $0<\theta<1$, are satisfied. Then, the system~\eqref{eq:plant} is ($\epsilon$,$\alpha$)-attackable with probability $\delta(T(\alpha+b+b_x,s_0),b_x,b_v)$ for some $\epsilon>0$, if  it holds that ${\phi}>{b}$ and $f\in \mathcal{U}_{\rho}$ with $\rho=2L_f(b_x+b+b_{\zeta})$, $\Vert s_0\Vert \leq \phi$ and  $b=\frac{c_4}{c_3-(L_{1}+L_{3}\Vert B\Vert)c_4}\sqrt{\frac{c_2}{c_1}}\frac{L_{2}b_{\zeta}}{\theta}$.
\end{theorem}


The theorem proof is provided in Appendix~\ref{app:t2}.

From~\eqref{eq:EXP_stable}, $c_3$ can be viewed as a `measure' of the closed-loop system stability (larger $c_3$ means the system converges faster to the equilibrium point); on the other hand, from Theorem~\ref{thm_closeloop}, closed-loop perception-based systems with larger $c_3$ are more vulnerable to stealthy attacks as the conditions of the theorem are easier to satisfy. However, if the plant's dynamics is very unstable, $T(\alpha+b_x+b,s_0)$ is smaller for a fixed $\alpha$ and $s_0$. Thus, the probability of attack success $\delta(T(\alpha+b_x+b,s_0),b_x,b_v)$ is larger for a fixed $b_x$ and $b_v$. 

It should be  further noted that $L'_f$, used in Theorem~\ref{thm_closeloop} conditions, would be equal to zero for LTI systems; thus, causing $L_1$ to become zero. Similarly when the mapping $\Pi$ approaches more towards linear behaviour, $L_{3}$ will go to zero. Therefore, the inequality $L_1+L_3 \Vert B\Vert<\frac{c_3}{c_4}$  holds for linear systems and linear controllers. However, larger values of $c_3$ will help the inequality to be satisfied even for nonlinear control systems. In simulation results described in Section~\ref{sec:simulation}, we discuss when these conditions are satisfied in more detail.

Moreover, in the extreme case when $b_\zeta=0$ (i.e., the attacker can exactly estimate the plant state), the condition $L_{2}b_{\zeta}<\frac{c_3-(L_{1}+L_{3}\Vert B\Vert)c_4}{c_4}\sqrt{\frac{c_1}{c_2}}\theta d$ will be relaxed and the other condition $L_{1}+L_{3}\Vert B\Vert < \frac{c_3}{c_4}$ becomes less restrictive as $L_1$ becomes smaller. Thus, 
in this case, if the attacker initiate the attack with arbitrarily small $s_0$, then \emph{$\epsilon$ can be arbitrarily close to zero and the attack will be very close to being strictly stealthy}. Hence, the following result holds. 

\begin{corollary}\label{cor:1}
Assume $b_\zeta=0$, $L_{1}+L_{3}\Vert B\Vert < \frac{c_3}{c_4}$ with $L_{1}=L'_f b_x$, $L_{3}=L'_{\Pi}(b_x+b_v)$ and the functions $f$, $f'$ (derivative of $f$) and $\Pi'$ (derivative of $\Pi$) are Lipschitz, with constants $L_f$, $L'_f$ and $L'_{\Pi}$, respectively. If $f\in \mathcal{U}_{\rho}$ with $\rho=2L_f(b_x+\Vert s_0\Vert)$, and $\Vert s_0\Vert \leq \frac{c_3-c_4(L_{1}+L_{3}\Vert B\Vert)}{c_4(L'_f+\Vert B\Vert L'_{\Pi})}$ holds, then  the system~\eqref{eq:plant} is ($\epsilon$,$\alpha$)-attackable with probability $\delta(T(\alpha+b_x+\Vert s_0\Vert,s_0),b_x,b_v)$, where $\epsilon=\sqrt{1-e^{-b_{\epsilon}}}$ for $b_{\epsilon}=\Big(\lambda_{max}(\Sigma_w^{-1})+\lambda_{max}(C_s^T\Sigma^{-1}_v C_s+\Sigma^{-1}_w)\times\min\{T(\alpha+b_x+\Vert s_0\Vert,s_0),\sqrt{\frac{c_2}{c_1}}\frac{e^{-\beta}}{1-e^{-\beta}} \}\Big)\Vert s_0\Vert^2$ and some  $\beta>0$. 
\end{corollary}

Finally, the above results depend on determining $\rho$ such that $f\in \mathcal{U}_{\rho}$. Hence, the following result provides a sufficient condition for $f\in \mathcal{U}_{\rho}$.

\begin{proposition}
Let $V: \mathbb{R}^n\to \mathbb{R}$ be a continuously differentiable function satisfying $V(0)=0$ and  define $U_{r_1}=\{x\in  B_{r_1}\,\,|\,\,V(x)>0\}$. Assume that $\Vert \frac{\partial{V}(x)}{\partial{x}}\Vert\leq \beta(\Vert x\Vert)$, and for any $x\in U_{r_1} $ it holds $V(f(x))-V(x)\geq \alpha(\Vert x\Vert)$, where $\alpha(\Vert x\Vert)$ and $\beta(\Vert x\Vert)$ are in class of $\mathcal{K}$ functions~\cite{khalil2002nonlinear}. Further, assume that $r_1$ can be chosen arbitrarily large. 
Now, 
\begin{itemize}
    \item 
    if $\lim_{\Vert x \Vert\to \infty}\frac{\alpha(\Vert x\Vert)}{\beta(\Vert x\Vert)}\to \infty $, then $f\in \mathcal{U}_{\rho}$ for any $\rho>0$. 
    \item 
    However, if $\lim_{\Vert x \Vert\to \infty}\frac{\alpha(\Vert x\Vert)}{\beta(\Vert x\Vert)}= \gamma$, then $f\in \mathcal{U}_{\rho}$ for any $\rho< \gamma$.
\end{itemize}
\end{proposition}

\begin{proof}
We prove the first case and the second case can be shown in a similar way. 

If $\lim_{\Vert x \Vert\to \infty}\frac{\alpha(\Vert x\Vert)}{\beta(\Vert x\Vert)}\to \infty $, there exists a bounded ball with radius $r_2$  and center at zero (referred to as $B_{r_2}$) such that for all $x\in S$ with $S=\{U_{r_1}-B_{r_2}^o\}$ it holds that $\frac{\alpha(\Vert x\Vert)}{\beta(\Vert x\Vert)} > \rho$.  Since the function $V$ is differentiable, using the Mean-value theorem for the dynamics $x_{t+1}=f(x_t)+d_t$ with $\Vert d_t\Vert \leq \rho$ and for any $x_t\in U_{r_1}$, we have~that
\begin{equation}
\begin{split}
V(x_{t+1})-V(x_t)=&V(f(x_t)+d_t)-V(x_t)\\
=&V(f(x_t))+(\frac{\partial{V}(x)}{\partial{x}})^Td_t-V(x_t)
\\\geq &\alpha(\Vert x\Vert)-\beta(\Vert x\Vert)\rho.
\end{split}
\end{equation}

Thus, for any $x_t\in S$, we have $V(x_{t+1})-V(x_t)>0$. Let us define $\eta=\min\{V(x_{t+1})-V(x_t)\,\,|\,\,x_t\in S \,\,\text{and}\,\, V(x_t)\geq a_{r_2}\}$ where $a_r=\max_{x\in{\partial B_r}}V(x)$ for any $r>0$. Such minimum exists as the considered set is compact and we have $\eta>0$ in $U_{r_1}$. Let us also assume $x_0=\{x\in B_{r_2}\,\,|\,\,V(x)=a_{r_2} \}$. Now, we claim that the trajectories starting from $x_0$ should leave the set $U_{r_1}$ through the boundaries of $B_{r_1}$. 

To show this, we know that for any $x_t\in S$, $V(x_t)\geq a_{r_2}$, since $V(x_{t+1})-V(x_t)\geq\eta>0$. Then, for any $t>0$
\begin{equation}
V(x_t)\geq V(x_0)+t\eta=a_{r_2}+\eta t.
\end{equation}
The above inequality shows that $x_t$ cannot stay in the set $S$ forever as $V(x)$ is bounded on the compact set $S$. On the other hand, $x_t$ cannot leave the set $S$ through the boundaries satisfying $V(x)=0$ or the surface of $B_{r_2}$ because  $V(x_t)>a_{r_2}$. Thus, the trajectories should leave the set $S$ through the surface of $B_{r_1}$, and as $r_1$ can be chosen arbitrarily large, the trajectories of $x_t$ will diverge to becoming arbitrarily large.
\end{proof}

\begin{example}\label{Ex:diverg}
Consider the dynamical system 
\begin{equation}
\begin{split}
x_{1,t+1}&=2x_{1,t}+x_{1,t}x_{2,t}^2\\
x_{2,t+1}&=0.5x_{2,t},\\
\end{split}
\end{equation}
and let us consider the function $V(x)=x_1^2-x_2^2$. 
Hence, we have $\Vert \frac{\partial{V}(x)}{\partial{x}}\Vert\leq 2\Vert x\Vert=\beta(\Vert x\Vert)$ and 
\begin{equation*}
\begin{split}
V(f(x_{t}))-V(x_t)=&4x_{1,t}^2+x_{1,t}^2x_{2,t}^4+2x_{1,t}^2x_{2,t}^2-0.25x_{2,t}^2\\
-x_{1,t}^2+x_{2,t}^2=&3x_{1,t}^2+0.75x_{2,t}^2+x_{1,t}^2x_{2,t}^4+2x_{1,t}^2x_{2,t}^2\\
\geq & 0.75 \Vert x_t\Vert^2=\alpha(\Vert x\Vert)
\end{split}
\end{equation*}
Since $\lim_{\Vert x \Vert\to \infty}\frac{\alpha(\Vert x\Vert)}{\beta(\Vert x\Vert)}\to \infty $, for any $\rho>0$, there exists $x$ such that $V(x_{t+1})-V(x_t)>0$ for the dynamics $x_{t+1}=f(x_t)+d_t$ with $\Vert d\Vert \leq \rho$. The shaded region in Fig.~\ref{fig:example} shows the area where $V>0$ outside $B_{r_2}$ and inside of $B_{r_1}$. Moreover, any area outside of the ball $B_{r_2}$ satisfies $V(x_{t+1})-V(x_t)>0$ (i.e., $V(x_t)$ is increasing over time). 
Now, by denoting the state value at point $A$ with $x_A$, we have $V(x_A)=x_{1,A}^2>0$. As $V(x_t$) is increasing outside $B_{r_2}$ as $t$ increases, any trajectory starting at $A$ cannot leave the shaded region from the surface of corresponding to $V=0$ or $B_{r_2}$; it also cannot stay in the region for all $t$, as the trajectory evolves.
Thus, the system states will eventually  leave the shaded region, `exiting' from the surface of the ball $B_{r_1}$. 

\end{example}

\begin{figure}[!t]
\centerline{\includegraphics[width=.6\columnwidth]{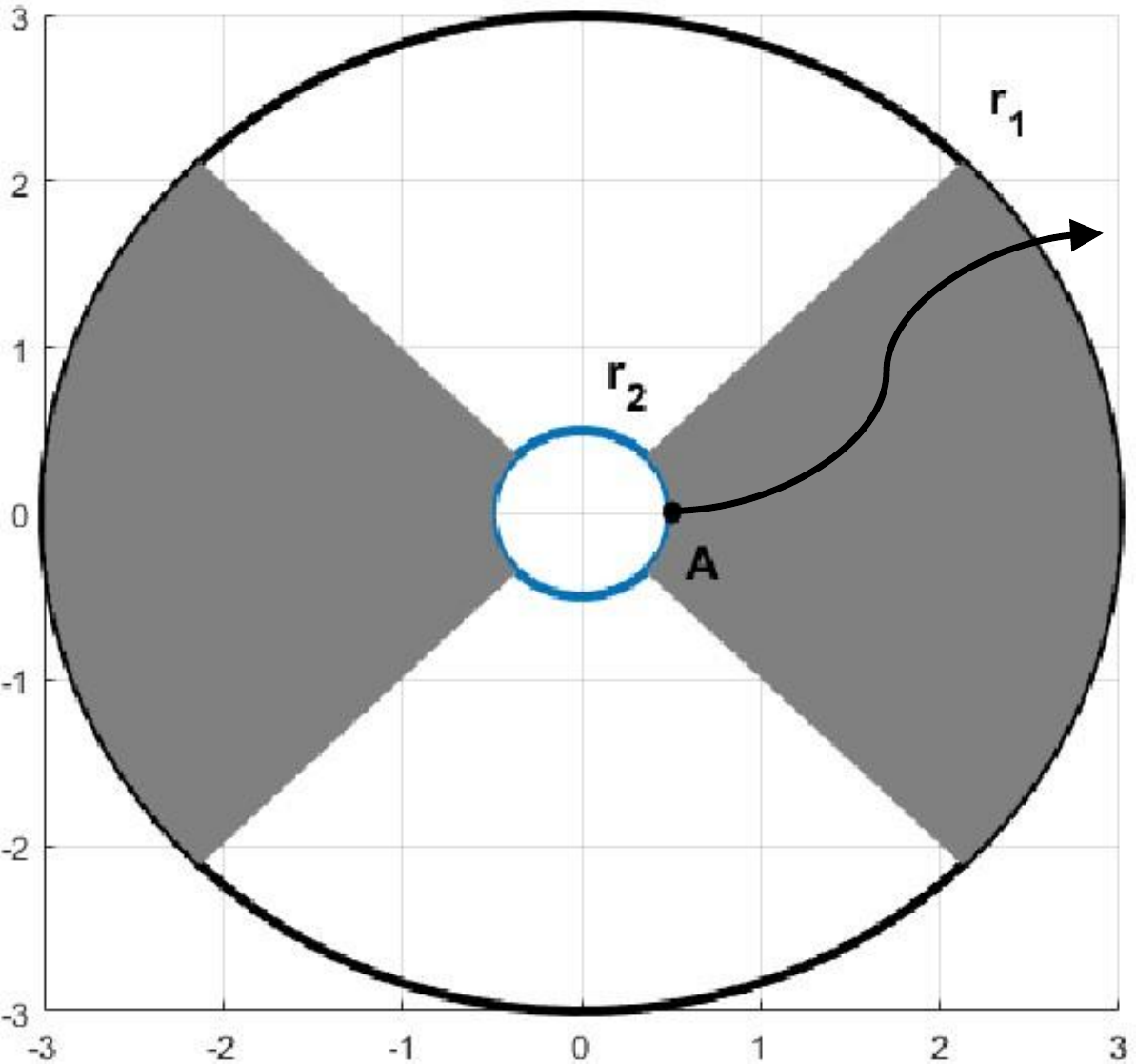}}
\caption{The trajectory of the dynamical system in Example~\ref{Ex:diverg}.}
\label{fig:example}
\end{figure}

Theorem~\ref{thm_closeloop} provides conditions on the closed-loop system stability that may not hold in general for any system with closed-loop exponential stability. In the following theorem, we show that one can still obtain a similar ($\epsilon$,$\alpha$)-successful attack (possibly with a larger $\epsilon$) even when the closed-loop system only satisfies exponential stability.

\begin{theorem}
Assume that the closed-loop control system~\eqref{eq:noiseless} is exponentially stable. Then, the system~\eqref{eq:plant} is ($\epsilon$,$\alpha$)-attackable with probability $\delta(T(\alpha+b+b_x,s_0),b_x)$ for some $\epsilon>0$, if $f\in \mathcal{U}_{\rho}$ with $\rho=2L_f(b_x+b+b_{\zeta})$ and  $b=\frac{c_4}{c_3}\sqrt{\frac{c_2}{c_1}}\frac{(L_f(2b_x+b_{\zeta})+2L_{\Pi}(b_x+b_{v}))}{\theta}$.
\end{theorem}

\begin{proof}
Using the same definition as in~\eqref{eq:controller} we have 
\begin{equation*}
\begin{split}
r_{t+1}=&h(r_t,0)+f(x_t^a)-f(x_t^a+\zeta_t)+f(e_t+\zeta_t)-f(x_t)\\&-f(r_t)
+B\Pi(e_t,v_{t}^{s})-B\Pi(x_t,v_{t}^{s})\\
&-B\Pi(r_t,0)=h(r_t,0)+\sigma_t,  
\end{split} 
\end{equation*}
From the Lipschitz property of function $f$, it follows that $\Vert f(x_t^a)-f(x_t^a+\zeta_t)\Vert \leq L_f b_{\zeta}$, and with probability $\delta(T(\alpha+b+b_x,s_0),b_x)$ we have 
\begin{equation*}
\begin{split}
\Vert f(x_t)\Vert &\leq L_f \Vert x_t\Vert \leq L_f b_x,\\
\Vert f(e_t+\zeta_t)-f(r_t)\Vert &\leq L_f \Vert x_t+\zeta_t\Vert \leq L_f(b_\zeta+b_x),\\
\Vert \Pi(x_t,v_{t}^{s})\Vert &\leq L_{\Pi} (\vert x_t\Vert+\Vert v_t^s\Vert)\leq L_{\Pi}(b_x+b_v), \\
\Vert \Pi(e_t,v_{t}^{s})-\Pi(r_t,0)\Vert  &\leq  L_{\Pi} (\vert x_t\Vert+\Vert v_t^s\Vert)\leq L_{\Pi}(b_x+b_v).
\end{split}
\end{equation*}
Therefore, with probability $\delta(T(\alpha+b+b_x,s_0),b_x)$ it holds that 
$\Vert \sigma_t\Vert \leq L_f(2b_x+b_{\zeta})+2L_{\Pi}(b_x+b_{v})$ for all $0\leq t\leq T(\alpha+b+b_x)$. 

Since the closed-loop system $r_{t+1}=h(r_t,0)$ is exponentially stable, we use Lemma~\ref{lemma:khalil} to show that the dynamics $r_{t+1}=h(r_t,0)+\sigma_t$ will remain in a bounded ball centered at zero with such probability  for all $0\leq t\leq T(\alpha+b+b_x)$. The bound is obtained by $b=\frac{c_4}{c_3}\sqrt{\frac{c_2}{c_1}}\frac{(L_f(2b_x+b_{\zeta})+2L_{\Pi}(b_x+b_{v}))}{\theta}$. Using the Data-processing inequality of KL divergence and following the same procedure as in Theorem~\ref{thm_closeloop} we obtain
\begin{equation*}
\begin{split}
 K&L\big(\mathbf{Q}({Y}_{-\infty}^{-1},Y_{0}^a:Y_{T(\alpha+b_x+b,s_0)}^a)||\mathbf{P}(Y_{-\infty}:Y_{T(\alpha+b_x+b,s_0)})\big)  \\
&\leq \sum_{i=0}^{T(\alpha+b_x+b,s_0)}\lambda_{max}(C_s^T\Sigma^{-1}_v C_s+\Sigma^{-1}_w)\Vert r_i \Vert^2\\
&\leq \sum_{i=0}^{T(\alpha+b_x+b,s_0)}\lambda_{max}(C_s^T\Sigma^{-1}_v C_s+\Sigma^{-1}_w)b=b_{\epsilon}\\
&\leq \lambda_{max}(C_s^T\Sigma^{-1}_v C_s+\Sigma^{-1}_w)b(T(\alpha+b_x+b,s_0)+1)=b_{\epsilon}.
\end{split}
\end{equation*}
This means that the system is ($\epsilon$,$\alpha$)-attackable with probability $\delta(T(\alpha+b_x+b,s_0),b_x,b_v)$ and $\epsilon=\sqrt{1-e^{-b_{\epsilon}}}$. On the other hand, since we have $f\in \mathcal{U}_{\rho}$ with $\rho=2L_f(b_x+b+b_{\zeta})$,  similarly as in Theorem~\ref{thm_closeloop}, we can show that for $t\geq T(\alpha+b+b_x,s_0)$ the states of the under-attack system will satisfy $\Vert x_t^a\Vert\geq \alpha$ with probability  $\delta(T(\alpha+b+b_x,s_0),b_x)$.
\end{proof}


Note that our results only focus on the existence of perception measurements $G(x_t^a-s_t)$, obtained by shifting the current perception scene by $s_t$, that results in $(\epsilon,\alpha)$-successful attack, 
and not how 
to compute it.
Further, 
to derive  attack sequence using Attack Strategy~${\textit{\rom{1}}}$, the attacker needs the estimation of the plant states. Thus, Attack Strategy~${\textit{\rom{2}}}$ relaxes this assumption, with the attacker only needing to have knowledge about the plant's (open-loop) dynamics $f$ and the computation power to calculate $s_{t+1}=f(s_t)$ ahead of time.

\subsection{Attack Strategy~${\textit{\rom{2}}}$: Using Plant Dynamics}

Similarly to Attack Strategy~${\textit{\rom{1}}}$, consider the attack sequence where
$z^a_t$ and $y_{t}^{s,a}$, for all $t\geq 0$,  
satisfy 
\begin{equation}
\begin{split}
z^a_t&=G(x_t^a-s_t), \,\,\, y_{t}^{s,a}=C_s(x_t^a-s_t)+v_t^{s,a},\\
s_{t+1}&=f(s_{t}),
\end{split}
\end{equation}
%
for some nonzero $s_{0}$. 
However, here the attacker does not need 
an estimate of the plant's state; they
simply 
follow plant dynamics $s_{t+1}=f(s_{t})$ to find the desired measurements' transformation.
Now, we define the state $e_t\delequal x_t^a-s_t$ as the \emph{fake state}, and the attacker's intention is to make the system believe that the plant state is equal to $e_t$.

The following theorem captures the condition for which the system is not resilient to attack strategy~${\textit{\rom{2}}}$, in the sense that it is $(\epsilon,\alpha)$-successful attackable.

\begin{theorem} \label{thm_openloop}
Consider the system~\eqref{eq:plant} that satisfies  Assumption~\ref{ass:control}. Assume that both functions $f'$ and $\Pi'$ are Lipschitz, with constants $L'_f$ and $L'_{\Pi}$, respectively. Moreover, assume that there exists $\phi>0$ such that the inequalities $L_{1}+L_{3}\Vert B\Vert < \frac{c_3}{c_4}$ and $L_2b_{x}<\frac{c_3-(L_{1}+L_{3}\Vert B\Vert)c_4}{c_4}\sqrt{\frac{c_1}{c_2}}\theta d$ with $0<\theta<1$ are satisfied, where $L_{2}=L'_f(\alpha+b_x)$, $L_{1}=L'_f(\alpha+\phi)$ and $L_{3}=L'_{\Pi}(b_x+\phi+b_v)$. Then, the system~\eqref{eq:plant} is ($\epsilon$,$\alpha$)-attackable with probability $\delta(T(\alpha+b_x+b,s_0),{b_x})$ and $b=\frac{c_4}{c_3-(L_{1}+L_{3}\Vert B\Vert)c_4}\sqrt{\frac{c_2}{c_1}}\frac{L_{2}b_{x}}{\theta}$, for  some $\epsilon>0$, if it holds that $\frac{\phi}{b}>1$ and $f\in \mathcal{U}_{0}$.
\end{theorem}

The proof of the theorem is provided in Appendix~\ref{app:t4}.

Unlike in Theorem~\ref{thm_closeloop}, $L_1$ and $L_2$ in Theorem~\ref{thm_openloop}
increase as $\alpha$ increases. Therefore, unless $L'_f=0$, one cannot claim that the attack can be $\epsilon$-stealthy for arbitrarily large $\alpha$ as the inequality $L_1+L_3\Vert B\Vert \leq \frac{c_3}{c_4}$ might not be satisfied. Therefore, there is a trade-off between the stealthiness  guarantees ($\epsilon$) and the performance degradation 
caused by~the~attack ($\alpha$). However, in an extreme case where the system is linear, it holds that $L'_f=0$ and 
$L_1=L_2=0$; before introducing the results for LTI systems in the next subsection, we remark on the following.

\begin{remark}
The fact that we considered the control input in a general end-to-end form can help us verify that our results also hold when the perception-based control system is modular; i.e., when the 
states are extracted from the perception module and physical sensors and then used with the classic control methods. This is because the end-to-end formulation of the controller is very general, also covering the modular methods. In such case, the traditional control methods can be used such that the noiseless closed-loop system becomes exponentially or asymptotically~stable.
\end{remark}

\subsection{Attack on LTI Systems}

To derive sufficient conditions for which  stealthy yet effective attacks exist, we have designed 
such attacks for perception-based control systems where plants have input affine nonlinear dynamics. However, for LTI plants, 
the obtained conditions 
can be significantly relaxed. 
Specifically, for 
LTI systems,  
because the system dynamics takes the form of $f(s_t)=As_t$,  Attack Strategies~${\textit{\rom{1}}}$ and ${\textit{\rom{2}}}$ become identical~as 
\begin{equation}\label{eq:attack_linear}
\begin{split}
s_{t+1}&=f(\hat{x}_t^a)-f(\hat{x}_t^a-s_t)=A(\hat{x}^a_t)-A(\hat{x}^a_t-s_t)\\
&=As_t=f(s_t).
\end{split}
\end{equation}
Therefore, 
the attacker 
does not need to estimate the state of the system and they can use the Attack Strategy~${\textit{\rom{2}}}$ to design the attack sequence. 

In general, assume that for such systems we use the controller satisfying Assumption~\ref{ass:control}, where due to complexity of perception (e.g., image observations), a nonlinear function maps the perception sensing (e.g., the image) and sensors to the control input.  Hence, 
from Corollary~\ref{cor:1}, we directly obtain the following result.

\begin{corollary}\label{corr:lti}
Consider an LTI perception-based control system with $f(x_t)=Ax_t$. 
Assume that $L_{3}\Vert B\Vert < \frac{c_3}{c_4}$ with $L_{3}=L'_{\Pi}(b_x+b_v)$,  $\Vert s_0\Vert \leq \frac{c_3-c_4L_{3}\Vert B\Vert}{c_4\Vert B\Vert L'_{\Pi}}$, and the matrix $A$ is unstable. Then
the system is ($\epsilon$,$\alpha$)-attackable with probability $\delta(T(\alpha+b_x+\Vert s_0\Vert,s_0),{b_x},b_v)$,  for arbitrarily large $\alpha$ and $\epsilon=\sqrt{1-e^{-b_{\epsilon}}}$, with $b_{\epsilon}=\Big(\lambda_{max}(\Sigma_w^{-1})+\lambda_{max}(C_s^T\Sigma^{-1}_v C_s+\Sigma^{-1}_w)\times\min\{T(\alpha+b_x+\Vert s_0\Vert,s_0),\sqrt{\frac{c_2}{c_1}}\frac{e^{-\beta}}{1-e^{-\beta}} \}\Big)\Vert s_0\Vert^2$ and some  $\beta>0$. 
\end{corollary}


Note that even though the above corollary considers LTI plants, the above requirement 
$L_{3}\Vert B\Vert < \frac{c_3}{c_4}$ is due to the nonlinearity of  employed controllers. 
%

Both Theorem~\ref{thm_closeloop} and Theorem~\ref{thm_openloop} assume \emph{end-to-end} controller that directly maps the perception and sensor measurements to the control input. However, there are controllers that first extract the state information using the perception module $P$ and then use a feedback controller to find the control input (e.g.,~\cite{dean2020robust,dean2020certainty}. 
For instance, consider a linear feedback controller with gain $K=\begin{bmatrix}K_P&K_s\end{bmatrix}$, resulting in
\begin{equation}\label{eq:linear_control}
\begin{split}
u_t=KY_{t}&=K_sy^{c,s}_{t}+K_Py^{P}_{t}\\
&=KCx_t+K_Pv^P(x_t)+K_sv^s_t. 
\end{split}
\end{equation}
Applying the above control input to the attack-free system 
\begin{equation}\label{eq:linear_in_control}
\begin{split}
x_{t+1}&=Ax_t+Bu_t+w_t\\
&=(A+BKC)x_t+BK_Pv^P(x_t)+BK_sv^s_t.
\end{split}
\end{equation}

Let us assume that Assumptions~\ref{ass:perception} and~\ref{ass:control} still hold. It is easy to show that $(A+BKC)$ needs to be a stable matrix (i.e., all eigenvalues are inside the unit circle). Then, we obtained the following result. 

\begin{theorem}\label{thm:linear}
Consider perception-based control of an LTI plant with dynamics $f(x_t)=Ax_t$, controlled with a linear feedback controller from~\eqref{eq:linear_control}. 
If the matrix $A$ is unstable, 
the system is ($\epsilon$,$\alpha$)-attackable with probability one for arbitrarily large $\alpha$ and $\epsilon=\sqrt{1-e^{-b_{\epsilon}}}$, where 
\begin{equation}
b_{\epsilon}=\lambda_{max}(C_s^T\Sigma^{-1}_v C_s+\Sigma^{-1}_w)(\frac{2\gamma T(\alpha+R_{\mathcal{S}},s_0)+\Vert s_0\Vert}{1-\lambda_{max}(A+BKC)}),
\end{equation}
and $\lambda_{max}(A+BKC)$ is the largest eigenvalue of the matrix $A+BKC$.
\end{theorem}

\begin{proof}
Let us assume that the attack dynamics are generated using~\eqref{eq:attack_linear}; by defining $e_t=x_t^a-s_t$, we obtain $z_t^{c,a}=G(x_t^a-s_t)=G(e_t)$ and $y_t^{c,s,a}=y_t^{s,a}-C_ss_t=C_se_t+v_t^s$. Therefore, 
\begin{equation}
\begin{split}
e_{t+1} &= x_{t+1}^a-s_{t+1}=Ax_{t}^a+BKY_t-As_t\\
&=Ae_t+BK\begin{bmatrix}C_Pe_t+v^P(e_t)\\
C_se_t+v_t^s\end{bmatrix}\\
&=(A+BKC)e_t+BK_Pv^P(e_t)+BK_sv^s_t.
\end{split}
\end{equation}

The above dynamics follows the same dynamics as in~\eqref{eq:linear_in_control}. Thus, if the initial condition $s_0$ is chosen small enough, $e$ will remain in the set $\mathcal{S}$ and it would holds that $\Vert v^P(e_t)\Vert\leq \gamma$ for all $t\geq 0$. Now, by defining $r_t=e_t-x_t$, we have
\begin{equation}
r_{t+1} = (A+BKC)r_t+v^P(e_t)-v^P(x_t)
\end{equation}

Since we have $\Vert v^P(e_t)\Vert\leq \gamma$ and $\Vert v^P(x_t)\Vert\leq \gamma$, and the matrix $(A+BKC)$ is stable, we have 
\begin{equation*}
\begin{split}
\Vert r_t\Vert =\Vert (A+&BKC)^{t} r_0 + \\
 + &\sum_{i=0}^{t-1} (A+BKC)^{t-i-1} (v^P(e_i)-v^P(x_i))\Vert  \\ \leq \vert\lambda_{max}&(A+BKC)\vert^t \Vert s_0\Vert + \frac{2\gamma}{1-\lambda_{max}(A+BKC)},
\end{split}
\end{equation*}
where we used $r_0=s_0$ and the squared matrix property of  $\Vert A^tv\Vert\leq \vert\lambda_{max}(A)\vert^t\Vert v\Vert$ for any $v\in \mathbb{R}^n$ and $A\in \mathbb{R}^{n\times n}$.
From the proof of Theorem~\ref{thm_closeloop}, we obtain
\begin{equation*}
\begin{split}
 KL&\big(\mathbf{Q}(Y_{0}^a:Y_{T(\alpha+R_{\mathcal{S}},s_0)}^a)||\mathbf{P}(Y_{0}:Y_{T(\alpha+R_{\mathcal{S}},s_0)})\big) \\
&\leq \sum_{i=0}^{T(\alpha+R_{\mathcal{S}},s_0)}\lambda_{max}(C_s^T\Sigma^{-1}_v C_s+\Sigma^{-1}_w)\Vert r_i \Vert^2\\
&= \lambda_{max}(C_s^T\Sigma^{-1}_v C_s+\Sigma^{-1}_w)\sum_{i=0}^{T(\alpha+R_{\mathcal{S}},s_0)}\Vert r_i \Vert^2\\
&= \lambda_{max}(C_s^T\Sigma^{-1}_v C_s+\Sigma^{-1}_w)(\frac{2\gamma T(\alpha+R_{\mathcal{S}},s_0)+\Vert s_0\Vert}{1-\lambda_{max}(A+BKC)}).
\end{split}
\end{equation*}
Since the matrix $A$ is unstable, choosing $s_0=cq_i$ where $c>0$ is a scalar and $q_i$ is the unstable eigenvector with associated eigenvalue $\vert \lambda_i \vert>1$, results in $s_t=c\lambda_i^tq_i$; thus,  $s_t$ becomes arbitrarily large for large enough $t$. 

Now, as $T(\alpha+R_{\mathcal{S}},s_0)$ is the time such that $\Vert s_t\Vert\geq \alpha+R_{\mathcal{S}}$, for all $t\geq T(\alpha+R_{\mathcal{S}},s_0)$, it holds that 
\begin{equation}
\begin{split}
\Vert s_t\Vert-\Vert x_{t}^a \Vert &\leq \Vert x_{t}^a-s_t\Vert=\Vert e_t \Vert \leq R_{\mathcal{S}} \Rightarrow\\
&\Rightarrow \Vert x_{t}^a \Vert \geq  R_{\mathcal{S}}+\alpha-R_{\mathcal{S}}=\alpha,
\end{split}
\end{equation}
which concludes the proof. 
\end{proof}

Finally, note that for LTI plants, the attacker will be effective if and only if $s_0$ is not orthogonal to all unstable eigenvectors of the matrix $A$. Therefore, if $\gamma$ is small enough, by choosing such $s_0$ that is also arbitrarily close to zero, the attacker can be $\epsilon$-stealthy with $\epsilon$ being arbitrarily small. On the other hand, $\gamma $ is not controlled by the attacker (rather, it is property of the controller design) and if $\gamma$ is large, there may be no stealthiness guarantee for the attacker. While this implies that having large $\gamma$ improves resiliency of the systems, it should be noted that large $\gamma$ is not  desirable from the control's perspective, as it would degrade the control performance when the system is free of attack.

The instability condition for LTI plants in  Theorem~\ref{thm:linear} is inline with the results from~\cite{mo2010false,kwon2014analysis} focused on LTI systems with linear controllers and without perception-based control, as well as with the notion of perfect attackability introduced there for the specific ADs (i.e., $\chi^2$) considered in those works. On the other hand, in this work, we consider a general notion of stealthiness (being stealthy from any AD) and provide analysis for perception-based control systems.

\begin{remark}
It should be noted the results derived in this work are sufficient conditions for the existence of $(\epsilon,\alpha)$ attacks and it does not mean the systems whose system model parameters does not satisfy these conditions are secure against attacks. However, one can prove that open-loop stability of the system can provide security for the system, however, providing a formal analysis is beyond the scope of this work.    
\end{remark}





\section{Simulation Results}
\label{sec:simulation}

We illustrate and evaluate our methodology for vulnerability analysis of perception-based control systems 
on two case studies,
 inverted pendulum and autonomous vehicles~(AVs). 
\subsection{Inverted Pendulum}
We consider a fixed-base inverted pendulum equipped with an end-to-end controller 
and a perception module that estimates the pendulum angle 
from camera images.
By using $x_1=\theta$ and $x_2=\dot{\theta}$, 
the inverted pendulum dynamics can be modeled in the state-space~form as
\begin{equation} \label{eq:state_Nonlin}
\begin{split}
\dot{x}_1&=x_2\\
\dot{x}_2&=\frac{g}{r}\sin{x_1} - \frac{b}{mr^2} x_2+ \frac{L}{mr^2};
\end{split}
\end{equation}
here, $\theta$ is the angle of the
pendulum rod from the vertical axis measured clockwise, $b$ is the Viscous friction coefficient, $r$ is the radius of inertia of the pendulum about the fixed point, 
$m$ is the mass of the pendulum, $g$ is the acceleration due to gravity, and $L$ is the external torque that is applied at the fixed base 
\cite{formal2006inverted}. 
Finally, we assumed 
$m=0.2 Kg$, $b=0.1 \frac{m\times Kg}{s}$, $r=0.3m$ and discretized the model with $T_s=0.01 \,\,s$.

Using Lyapunov's indirect method one can show that the origin of the above system is unstable because the linearized model has an unstable eigenvalue. However, the direct Lyapunov method can help us to find the whole unstable region $-\pi <\theta<\pi$ (see~\cite{khalil2002nonlinear}). 
%
We used a data set $\mathcal{S}$ with $500$ sample pictures of the fixed-base inverted pendulum with different angles in $(-\pi,\pi)$ to train a DNN $P$ (perception module) to estimate the angle of the rod. 
The angular velocity is also 
measured directly by the sensor with noise variance of $\Sigma_v=0.0001$. 
We also trained a deep reinforcement learning-based controller directly mapping the image pixels and angular velocity values to the input control.  For anomaly detection, we designed a standard $\chi^2$ Extended Kalman filter AD that receives the perception module's P output and angular velocity, and outputs the residue/anomaly~alarm.

\begin{figure}[!t]

\centering
{\label{fig:theta_strategy1}
  \includegraphics[clip,width=0.48\columnwidth]{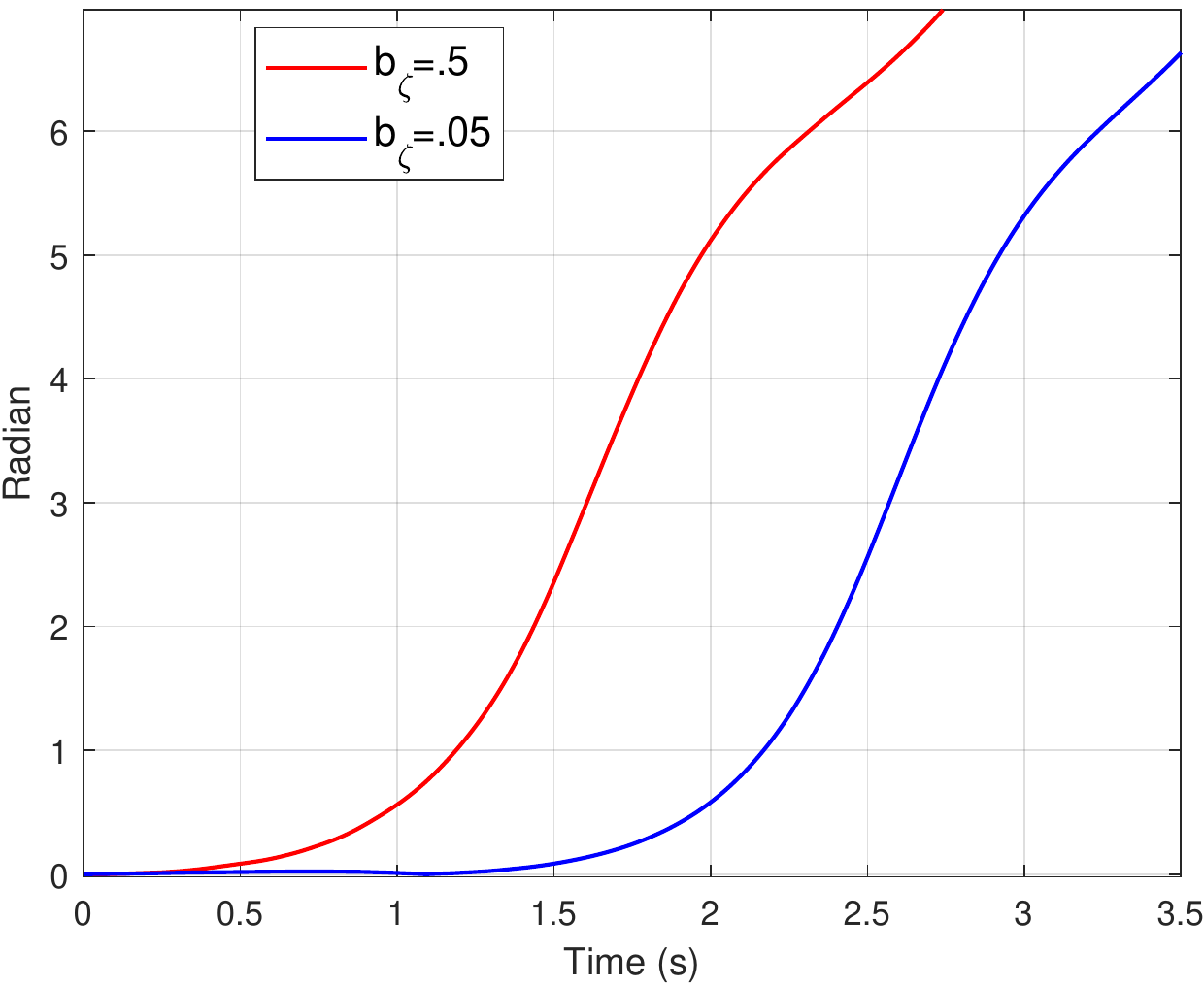}
}
{\label{fig:residue_strategy1}
  \includegraphics[clip,width=0.494\columnwidth]{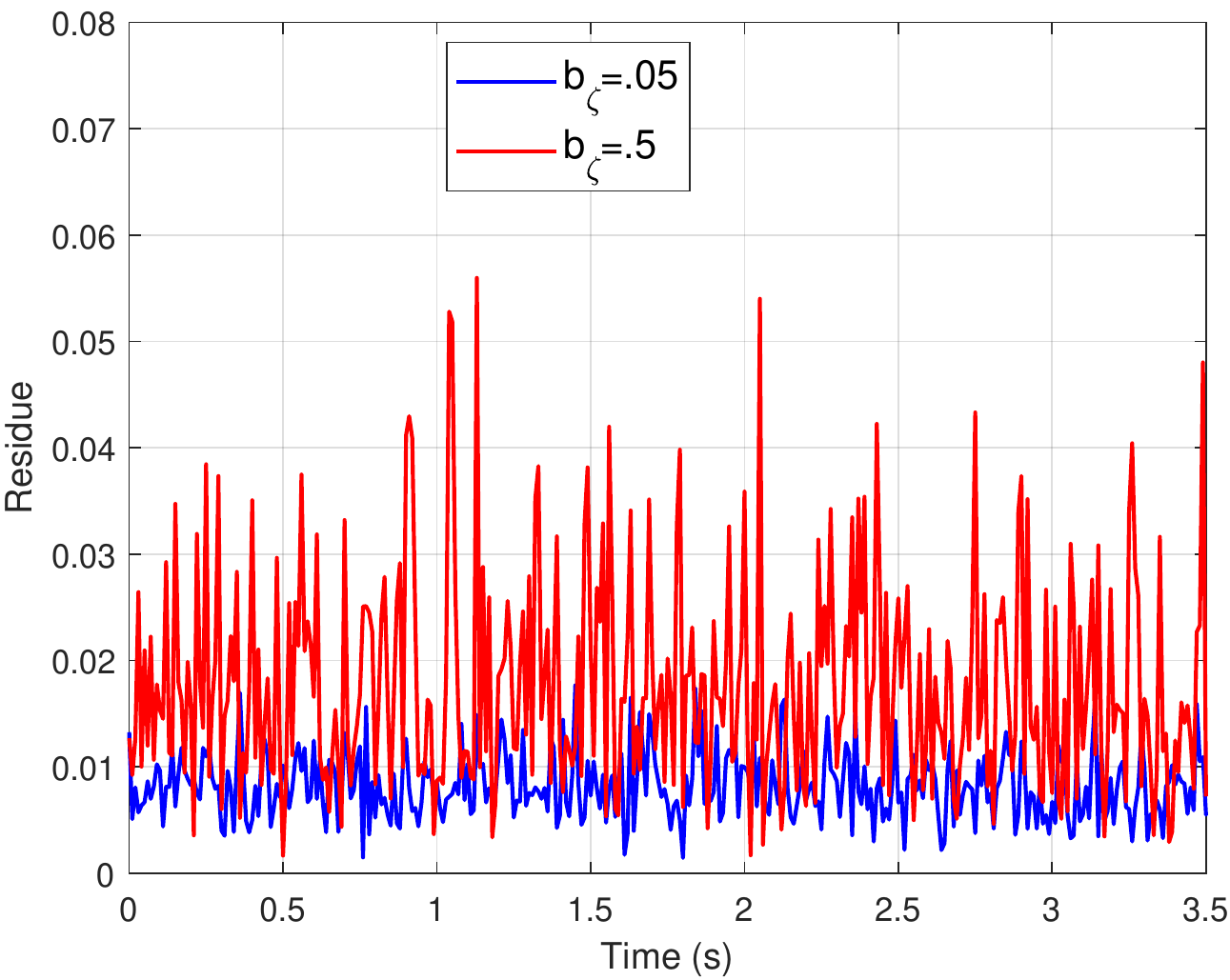}%
}

\caption{(a) Evolution of the angle's ($\theta$) absolute value over time for different levels of $b_{\zeta}$; (b) The norm of the residue over time when the attack starts at time $t=0$. }
\label{fig:strategy_1}
\end{figure}

\begin{figure}[!t]
\centering
{\label{fig:theta_strategy2}
  \includegraphics[clip,width=0.476\columnwidth]{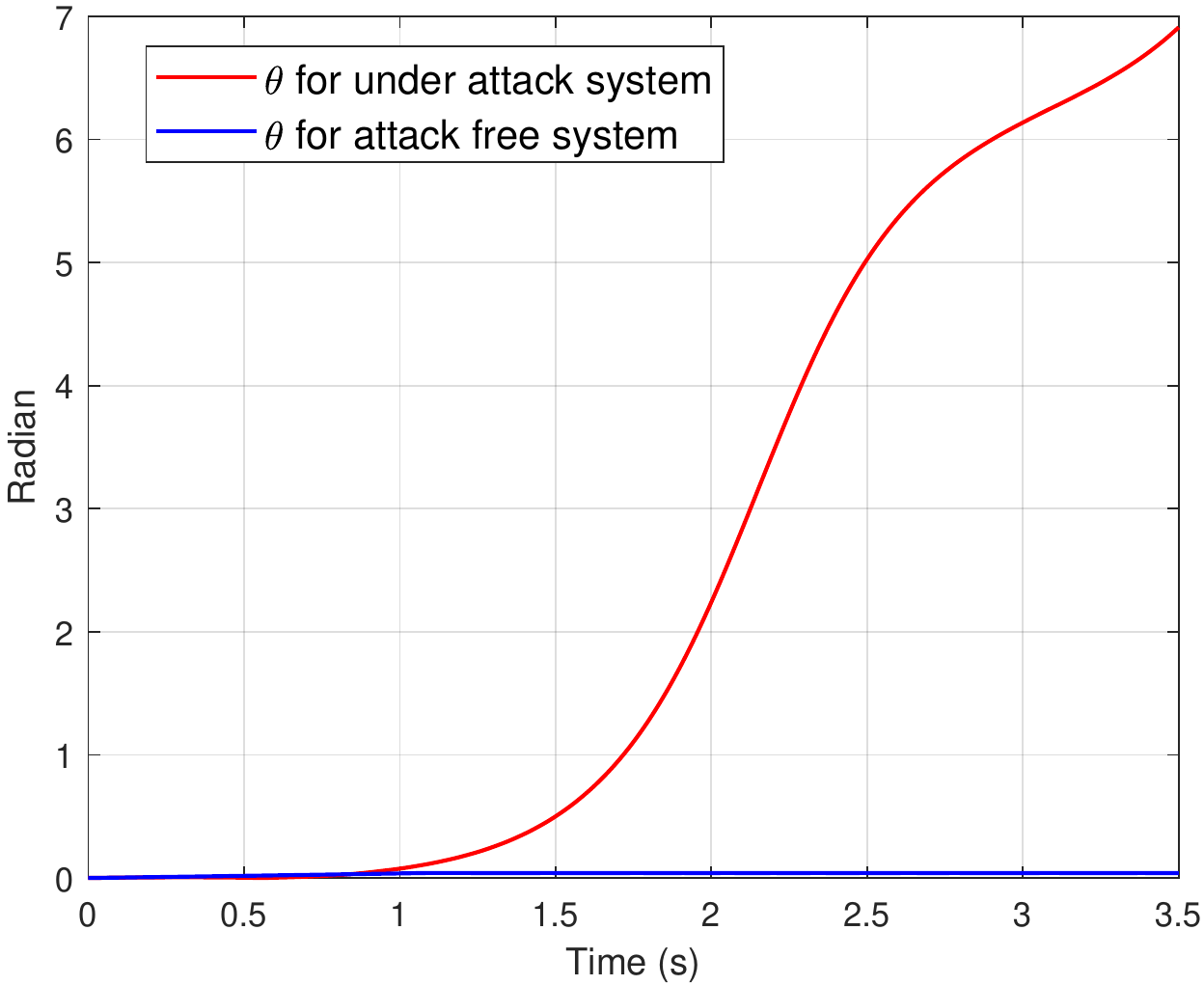}
}
{\label{fig:residue_strategy2}
  \includegraphics[clip,width=0.496\columnwidth]{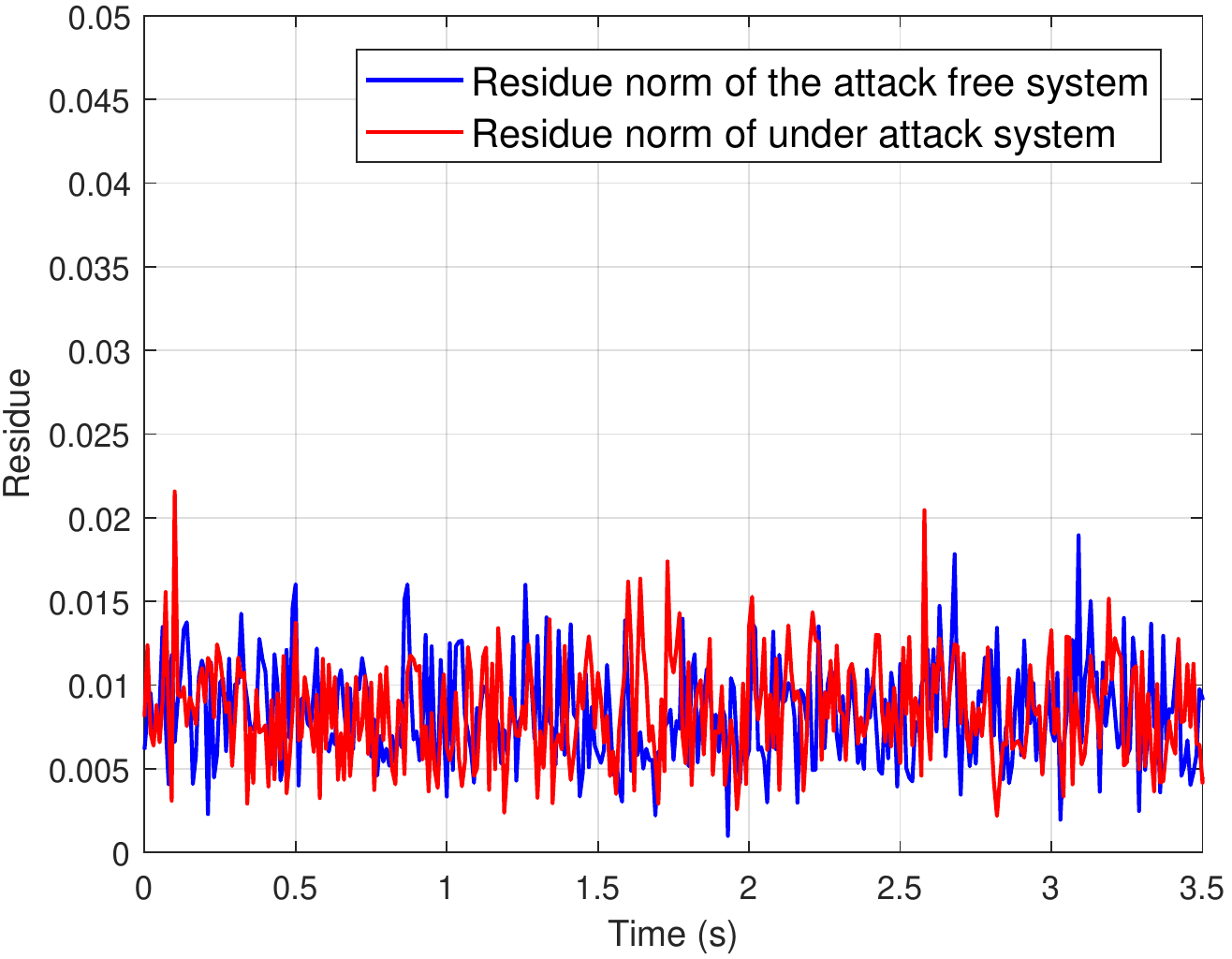}%
}

\caption{(a) Angle's ($\theta$) absolute value over time for Attack Strategy~${\textit{\rom{2}}}$ (red) and normal condition (blue); (b) The residue norm over time for both under-attack and attack-free systems. }\label{fig:strategy_2}

\end{figure}

After discretizing the model with sampling rate of $dt=0.01 sec$, one can verify that $L'_f=0.33$, $\Vert B \Vert=0.556$ and we also found that $L'_{\Pi}=0.12$ for the controller. We also considered a Lyapunov function with $c_1=c_2=0.5$, $c_3=0.057$ and $c_4=1$ that satisfies the inequalities from~\eqref{eq:EXP_stable}. On the other hand, in more than 100 experiments without the attack, each simulated for a period of 5 seconds (i.e., 500
time steps), 
we found that for $b_x=0.2 $ and $b_v=0.05 $ the inequalities in Definition~\ref{def:b_x} were always satisfied -- i.e., holding with probability of almost one. Now, one can verify that the condition in Theorem~\ref{thm_closeloop} is satisfied if $b_{\zeta}< .051$.


We first choose $s_{0}=\begin{bmatrix}
0.001& 0.001
\end{bmatrix}^T$ and used 
Attack Strategy~${\textit{\rom{1}}}$ to 
design attacks. To derive the current adversarial image at each time step, the attacker exploits the actual image and compromise it by deviating the pendulum rod by $s_t$ degrees. This compromised image is then delivered to the controller, to be used by the perception module to evaluate the system state. {The other sensor's measurements are also compromised accordingly.} Again, the attacker does not need to have access to the perception map $P$; the knowledge about the dynamics $f$ and estimate of the current plant state $\hat{x}_t^a$ is sufficient to craft the perturbed images. 

Fig.~\ref{fig:strategy_1}(a) shows the actual pendulum rod angle for different estimation uncertainty level $b_{\zeta}$ (by the attacker) when the attack starts at $t=0$. In both cases, the attacker can drive the pendulum rod into an unsafe region. Fig.~\ref{fig:strategy_1}(b) shows the residue signal over time; the attack stealthiness level decreases as $b_{\zeta}$ increases, consistent with our results in Section~\ref{sec:perfect}.

Fig.~\ref{fig:strategy_2}(b) presents the residue of the system in normal operating condition (i.e., without attack) 
as well as under Attack Strategy~${\textit{\rom{2}}}$. We can see that the residue level of both Attack Strategy~${\textit{\rom{1}}}$ with $b_{\zeta}=.05$ and Attack Strategy~${\textit{\rom{2}}}$ are the same as for the system without attack. The red and blue line in Fig.~\ref{fig:strategy_2}(a) also show the pendulum rod angle trajectory for Attack Strategy~${\textit{\rom{2}}}$ and normal condition, respectively.




 \subsection{Autonomous  Vehicle}

We consider the nonlinear dynamical model of an AV from~\cite{kong2015kinematic}, with four states 
$\left[x ~~~y~~~ \psi~~~v\right]^T$ in the form
\begin{equation}\label{eq:car}
\begin{split}
\dot{x}&=v\cos({\psi+\beta}),~~~~
\dot{y}=v\sin({\psi+\beta}),\\
\dot{\psi}&=\frac{v}{l_r}\sin(\beta),~~
\dot{v}=a,~~
\beta=\tan^{-1}(\frac{l_r}{l_f+l_r}\tan(\delta_f));
\end{split}
\end{equation}
here, $x$ and $y$ represent the position of the center of mass
in $x$ and $y$ axis, respectively, $\psi$ is the inertial heading, $v$
is the velocity of the vehicle, and $\beta$ is the angle of the current velocity of the center of mass with respect to the longitudinal axis of
the car. 
Also, $l_f=1.1$ and $l_r=1.73 m$ are the distance from the center of the mass of the vehicle to the front and
rear axles, respectively, and $a$ is the acceleration. The control inputs are the front steering angle $\delta_f$ and acceleration $a$. We assume only $\psi$ is measured directly using noisy sensors, 
with zero-mean noise with variance $\Sigma_v=0.0001$. Moreover, there is a camera affixed to dashboard that receives the image of the scene in front of the car. The system noise is assumed to be zero-mean, with covariance 
$\Sigma_w=0.0001$.

We consider the scenario where the car has a constant speed of $25 m/s$. Since the control objective 
is to keep the vehicle within the lane (i.e., between the lines of the road), the relative position of the camera with respect to the two lines is the essential information used by the controller. Therefore, even if the image has some other backgrounds, we assume that the lines of the road would be detected, to allow for extracting the relevant information about the states of the system (including the position of the car with respect to the lane center). Using deep reinforcement learning, we trained a 
controller that takes the images (containing the lines) and the inertial heading measurements $\psi$ as the observations, and maps them to the control input $\delta_f$ that keeps the car between the lanes. As the reward function, we assign a higher reward when the car is between the lines and a lower reward when the car is moving further from the center of the lane.  To find the position of the car with respect to the lane center, we trained another DNN (perception map) that takes the images and renders the position y in state space as follows
%
$P(z_t )=y_t+v^P(y_t).$

This information is used in AD, where we used $P(z_t )$ and measured $\psi$ to find the residue for both $\chi^2$ and CUSUM detectors. 
We assume the system is equipped with $\chi^2$\cite{jovanov_tac19} and CUSUM anomaly detectors~\cite{umsonst2017security}, where an Extended Kalman filter is used to find the residue signal. The thresholds are set to have $p^{FA}=0.05$ as false alarm rate. Although the system dynamics~\eqref{eq:car} is not in the form of~\eqref{eq:plant}, the car's kinematics in lateral movement 
can be approximated well by an LTI model with matrix $A$ around operating point $(x_e,0,0,0)$ as 
$A=\bigl[ \begin{smallmatrix}
1&0&0&dt\\0&1&25dt&0\\0&0&1&0\\0&0&0&1
\end{smallmatrix}\bigr];$
the model was obtained by linearizing the state dynamics  with sampling time of $dt=0.01 sec$, 
where $x_e$ is any arbitrary value. Then, we considered the LTI-based attack 
design using as $s_{t+1}=As_t$ for some small nonzero initial condition $s_0=0.001\begin{bmatrix}0&1&1&0\end{bmatrix}^T$. Since the open-loop model is linear we verify the conditions of Corollary~\ref{corr:lti}; the only condition that needs to be satisfied in this case is $L_3 \Vert B\Vert<c_3$. 

We obtained $L'_{\Pi}=.23$ and considered a Lyapunov function with $c_1=c_2=0.5$, $c_3=0.032$ and $c_4=1$ that satisfies the conditions from~\eqref{eq:EXP_stable}. Also, when running 100 experiments without the attack for, each for a simulation period of 30 seconds (i.e., 3000 time steps) in all experiments we obtained 
$\sup_{0\leq t\leq 3000}\Vert x_t\Vert <0.2,~~\sup_{0\leq t\leq 3000} \Vert v^s_t\Vert<0.05$,
%
which means that for $b_x=0.2$ and $b_v=0.05$, with probability of almost one the system stays in  
the zero center ball, as in Definition~\ref{def:b_x}.
Now, we have $L_3=L'_{\Pi}(b_x+b_v+\phi)=0.0532$ where the upper bound on the norm of the initial condition of attack states $s$ is considered to be $0.0014$. 
Thus, having $\Vert B\Vert=0.556$ (obtained by linearizing and discretization), 
\begin{equation*}
L_3 \Vert B\Vert =0.0106<.032.
\end{equation*}

Now, $y_t^a-s_{t,2}$ provides the desired position of the camera (center of the car) in Y-axis (where $s_{t,2}$ is the second element of $s_t$). To find the attacked image, we distorted the current image by shifting it some pixels to right or left (depending on sign of $s_{t,2}$) in order to have camera be placed in $y_t^a-s_{t,2}$. Another approach to find the attacked image is to use some prerecorded images of the road where the car is placed at with different distances from the lane center. Then, the goal would be to find an image from those pre-recorded images whose distance from the lane center is closest to $y_t^a-s_{t,2}$.


Fig.~\ref{fig:car_distance} shows the position of the center of the car with respect to the road center when the attack starts at time zero. If the attacker chooses the initial condition $s_0=0.001\begin{bmatrix}0&1&1&0\end{bmatrix}^T$ the car will deviate to the left side of the road (left figure), while choosing $s_0=-.001\begin{bmatrix}0&1&1&0\end{bmatrix}^T$ will push the car to the right hand side (right figure). 

Fig.~\ref{fig:car_alarm_rate} illustrates the average number of alarms at each time step for both $\chi^2$ (left figure) and CUSUM (right figure) anomaly detectors in 1000 experiments when the attack starts at time $t=0$. As shown, the values of true alarm averages (for $t>0$) are the same as the false alarm averages (for $t\leq 0$) for both ADs, which indicates the stealthiness of the attack according to Definition~\ref{def:stealthy}.

\begin{figure}[!t]

\centering
{\label{fig:car_distance1}
  \includegraphics[clip,width=0.48\columnwidth]{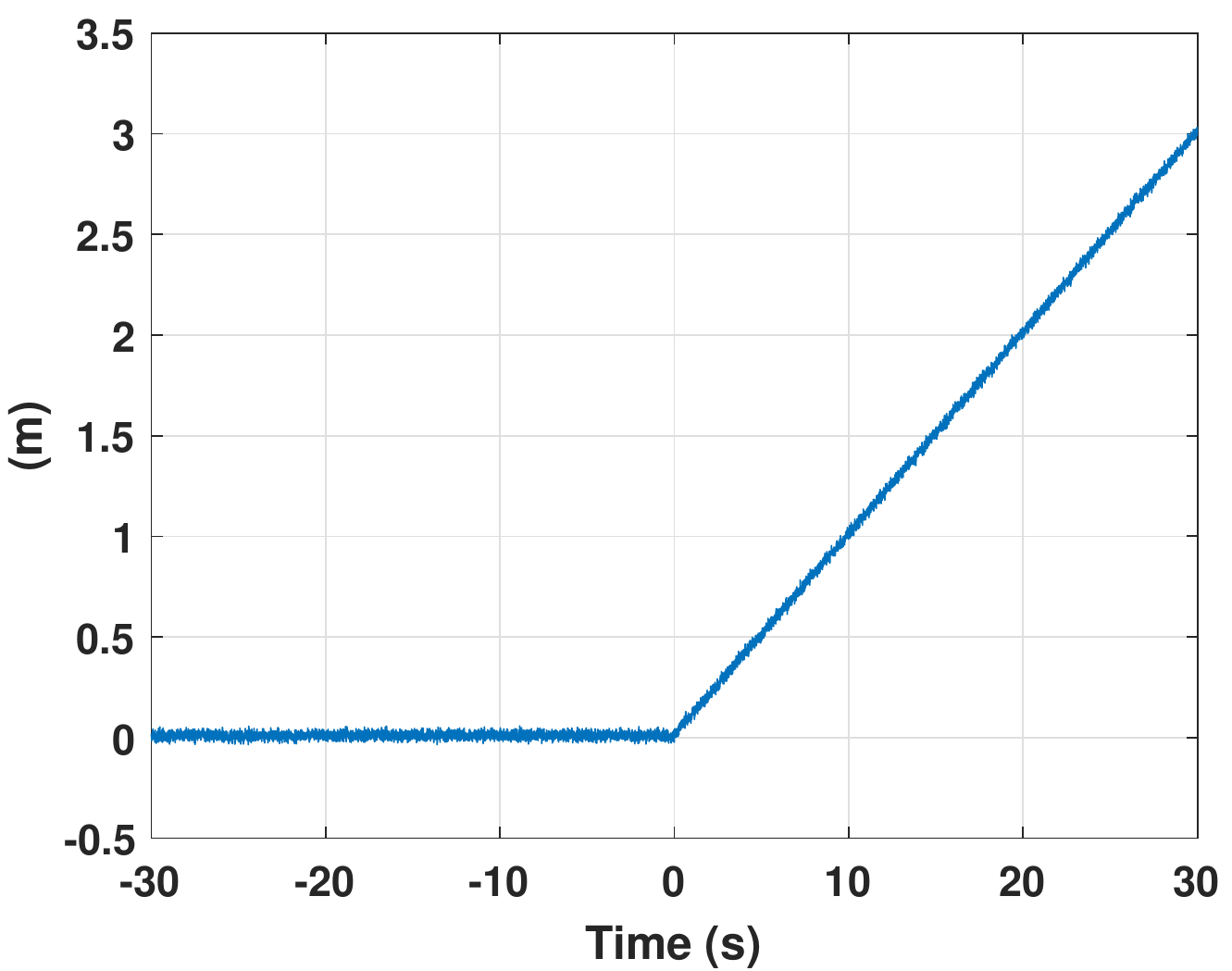}
}
{\label{fig:car_distance2}
  \includegraphics[clip,width=0.48\columnwidth]{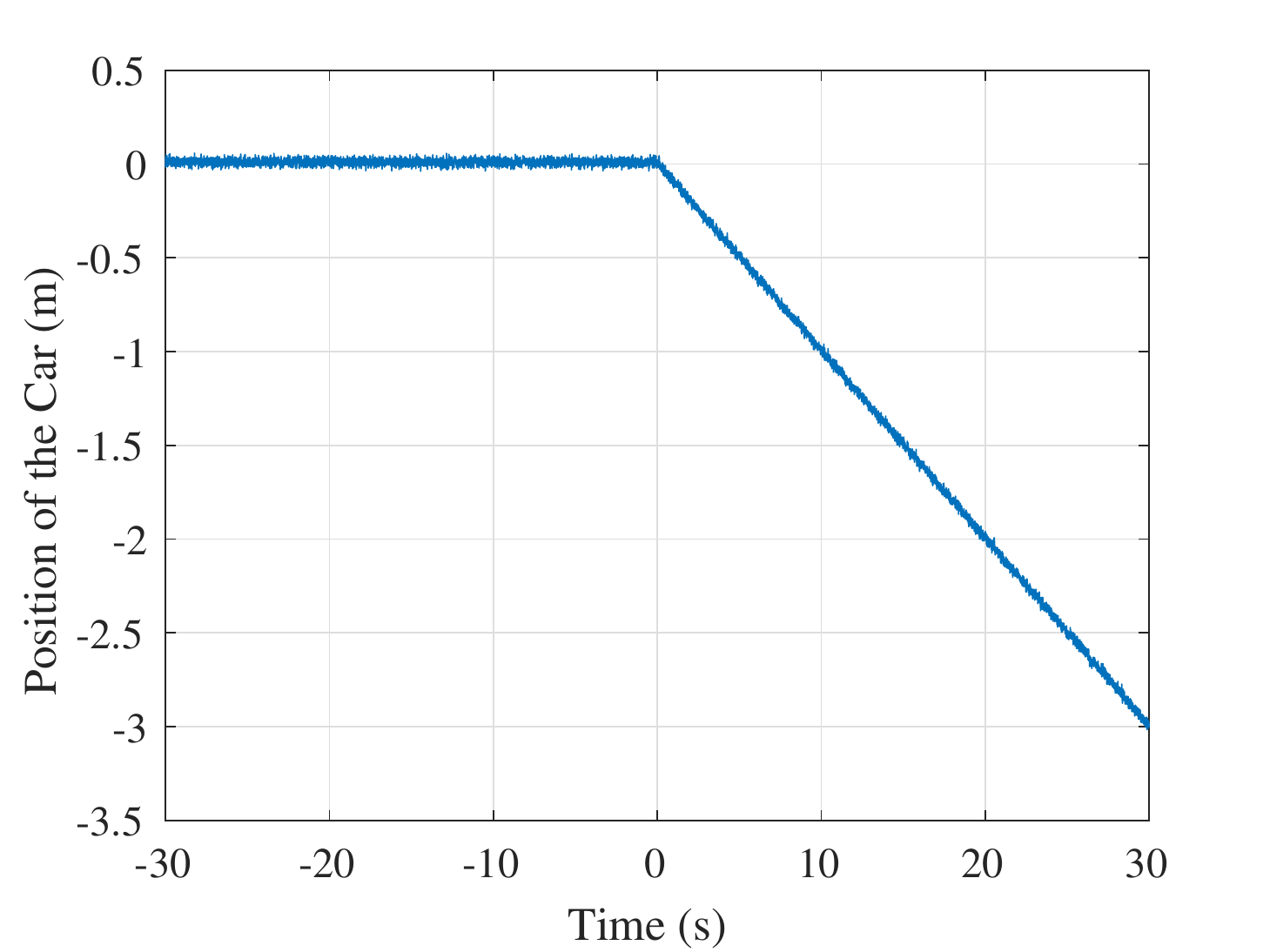}%
}

\caption{(left) The position of the car with respect to the road center over time for the initial condition $s_0=0.001\begin{bmatrix}0&1&1&0\end{bmatrix}^T$ when the attack starts at time $t=0$; (right) the position of the car with respect to the road center over time for the initial condition $s_0=-0.001\begin{bmatrix}0&1&1&0\end{bmatrix}^T$ when the attack starts at time $t=0$. }
\label{fig:car_distance}
\end{figure}

\begin{figure}[!t]
\centering
{\label{fig:theta_strategy1}
  \includegraphics[clip,width=0.482\columnwidth]{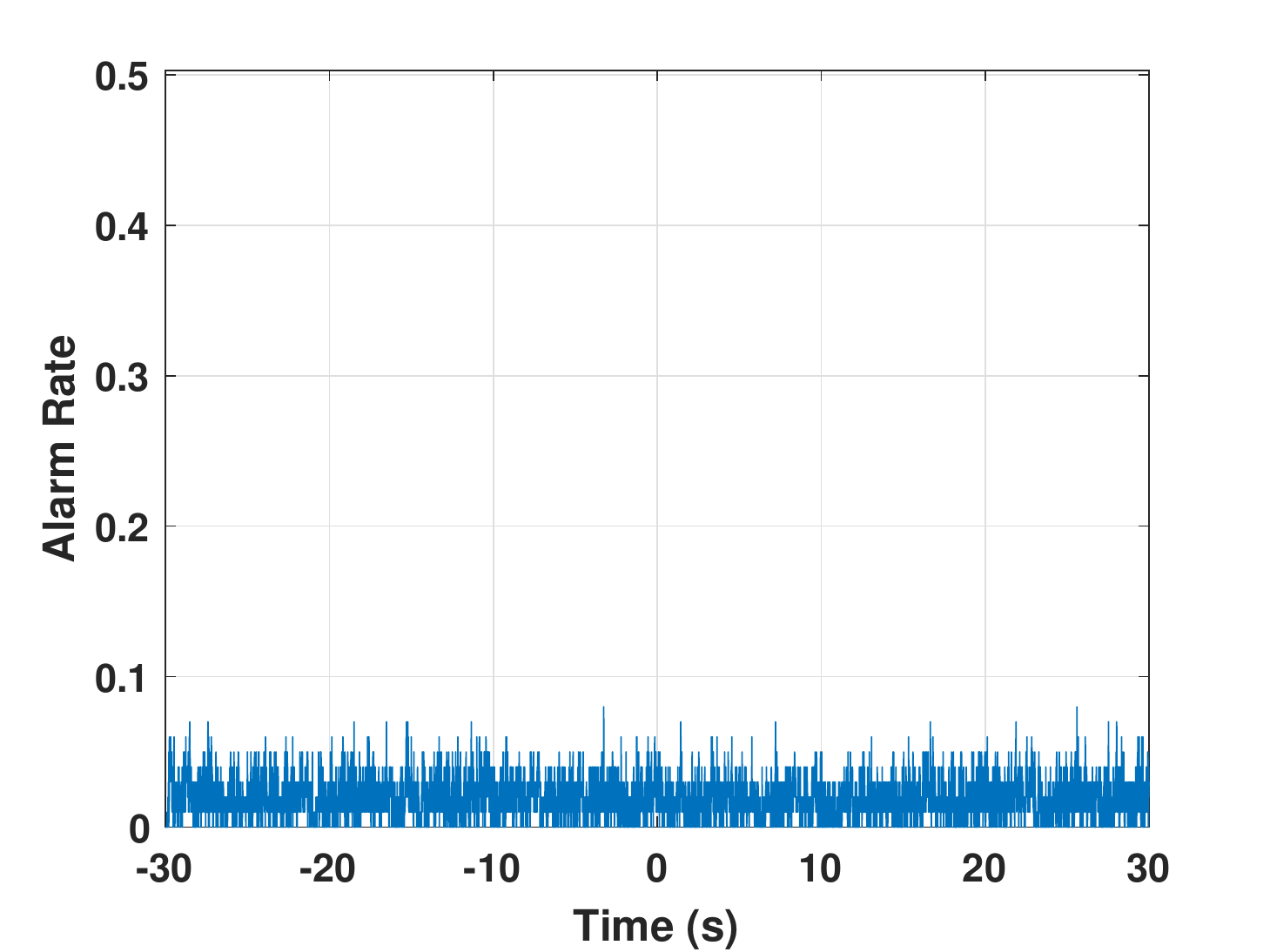}
}
{\label{fig:residue_strategy1}
  \includegraphics[clip,width=0.482\columnwidth]{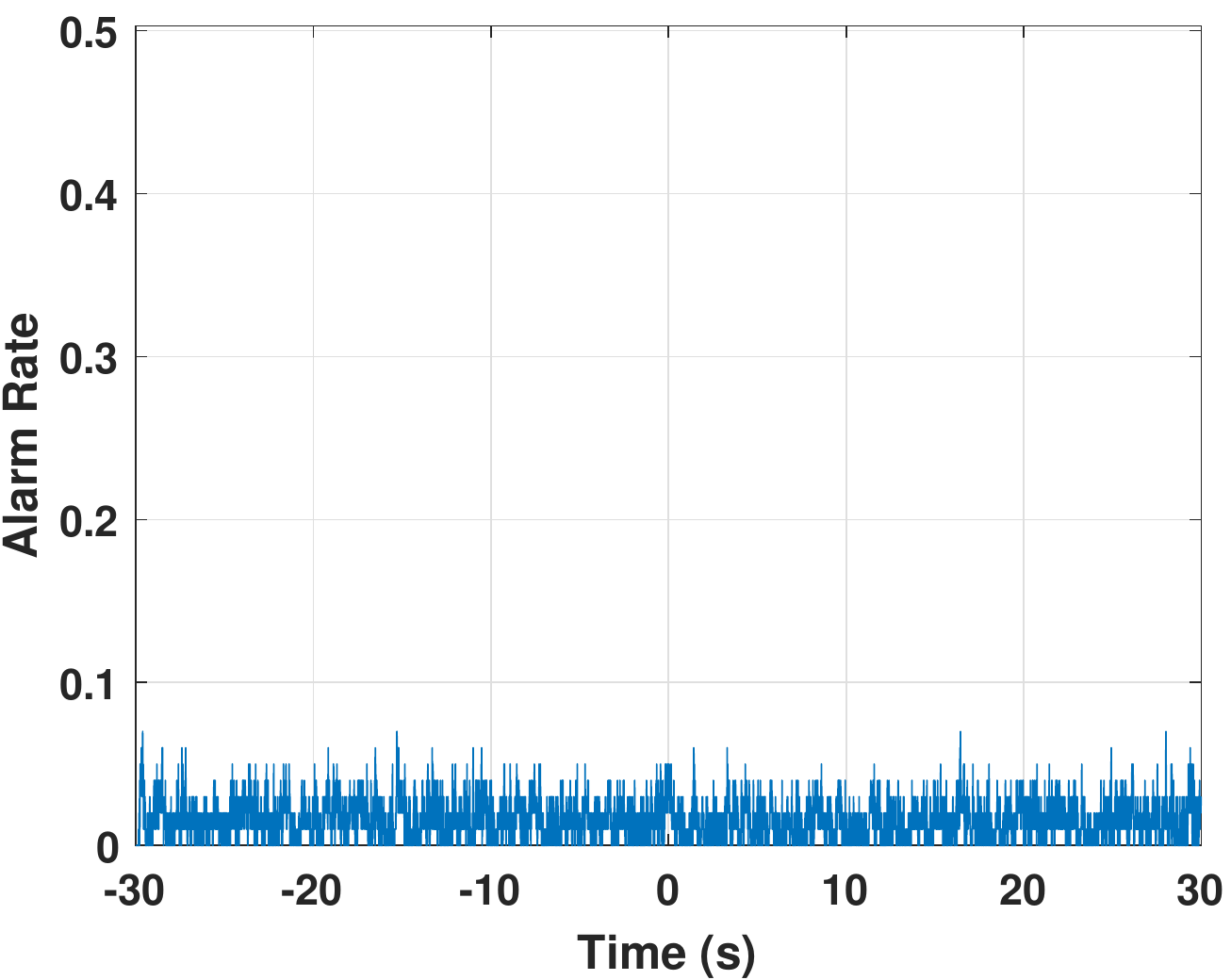}%
}
\caption{(left) The average number of alarms at each time step for $\chi^2$-based intrusion detector when the attack starts at time $t=0$; (right) the average number of alarms at each time step for CUSUM-based anomaly detector when the attack starts at time $t=0$. }
\label{fig:car_alarm_rate}
\end{figure}

\section{Conclusion}\label{sec:conclusion}
In this work, we have considered the problem of resiliency under sensing and perception attacks for perception-based control systems, focusing on a general class of nonlinear dynamical plants.  We have assumed that the noiseless closed-loop system equipped with an end-to-end controller and anomaly detector, is exponentially stable on a set around the equilibrium point. We have introduced a general notion of $\epsilon$-stealthiness as a measure of difficulty in attack detection from the set of perception measurements and sensor values. Further, we have derived sufficient conditions for an effective yet $\epsilon$-stealthy attack sequences to exist. 
Here, the control performance degradation has been considered as  moving the system state outside of the safe region defined by a bounded ball with radius $\alpha$, resulting in an $(\epsilon,\alpha)$-successful attack. 
Finally, we illustrated our results on two case studies,  fixed-base inverted pendulum and vehicle lane-keeping. 

\appendix

\subsection{Proof of Theorem~\ref{thm_closeloop}}
\label{app:t2}

\begin{proof}[Theorem~\ref{thm_closeloop}]
We need to show that the sequence of falsified perception (e.g., images) and physical sensor values $\{z_{0}^{c,a},y_{0}^{c,s,a}\},\{z_{1}^{c,a},y_{1}^{c,s,a}\},...$ obtained by Attack Strategy~${\textit{\rom{1}}}$ is ($\epsilon$,$\alpha$)-successful attack for arbitrarily large $\alpha$. By defining $e_t=x_t^a-s_t$ and $r_t=e_t-x_t$, we get $z^{c,a}_t=G(e_t)$, $ y_{t}^{c,s,a}=C_se_t+v_t^{s}$ and 
\begin{equation}\label{eq:e.r}
\begin{split}
e_{t+1}=&f(x_t^a)-f(x_t^a+\zeta_t)+f(e_t+\zeta_t)+B\Pi(e_t,v_{t}^{s})\\&+w_t^a,\\
r_{t+1}=&h(r_t,0)+f(x_t^a)-f(x_t^a+\zeta_t)+f(e_t+\zeta_t)-f(x_t)\\&-f(r_t)
+B\Pi(e_t,v_{t}^{s})-B\Pi(x_t,v_{t}^{s})-B\Pi(r_t,0)\\
=&h(r_t,0)+\sigma_1'+\sigma_2',
\end{split}
\end{equation}
with $\sigma_1'=f(x_t^a)-f(x_t^a+\zeta_t)+f(e_t+\zeta_t)-f(x_t)-f(r_t)$ and $\sigma_2'=B\Pi(G(e_t),y_{t}^{s,a})-B\Pi(G(x_t),y_{t}^{s})-B\Pi(G(r_t),C_sr_t)$. 

Using the Mean value theorem and equality $x_t^a=r_t+s_t+x_t$
we have that
\begin{equation}
\begin{split}
f(x_t^a+\zeta_t)&=f(x_t^a)+\frac{\partial f}{\partial x}|_{\Theta_{(x^a,x^a+\zeta)}}\zeta_t,\\
f(e_t+\zeta_t)&=f(x_t)+\frac{\partial f}{\partial x}|_{\Theta_{(x,x+r+\zeta)}}(r_t+\zeta_t),\\
f(r_t)&=f(0)+\frac{\partial f}{\partial x}|_{\Theta_{(0,r)}}r_t=\frac{\partial f}{\partial x}|_{\Theta_{(0,r)}}r_t,
\end{split}
\end{equation}
where for instance, $\frac{\partial f}{\partial x}|_{\Theta_{(x,y)}}=\begin{bmatrix}\nabla f_1\big(c_1x+(1-c_1)y\big)\\\vdots\\ \nabla f_n\big(c_n x+(1-c_n)y\big) \end{bmatrix}$ for some $0\leq c_1,...,c_n\leq 1$ where $f_i$ is the $i$-th element of function $f$. Therefore, we get that
\begin{equation}
\begin{split}
\sigma_1'=&(\frac{\partial f}{\partial x}|_{\Theta_{(x^a,x^a+\zeta)}}-\frac{\partial f}{\partial x}|_{\Theta_{(x,x+r+\zeta)}})\zeta_t+\\
&(\frac{\partial f}{\partial x}|_{\Theta_{(x,x+r+\zeta)}}-\frac{\partial f}{\partial x}|_{\Theta_{(0,r)}})r_t.
\end{split}
\end{equation}

Similarly, using the Mean Value theorem and $e_t=r_t+x_t$ we obtain
\begin{equation}
\begin{split}
\Pi(e_t,v_{t}^{s})&=\Pi(x_t,v_{t}^{s})+\frac{\partial \Pi}{\partial e_t}|_{\Theta_{(x,x+r)}} r_t\\
\Pi(r_t,0)&=\Pi(0,0)+\frac{\partial \Pi}{\partial r_t}|_{\Theta_{(0,r)}} r_t,
\end{split}
\end{equation}
where $\Pi(0,0)=0$. Therefore, we get
\begin{equation}
\sigma_2'=B(\frac{\partial \Pi}{\partial e_t}|_{\Theta_{((x,v),(x+r,v))}}-\frac{\partial \Pi}{\partial r_t}|_{\Theta_{((0,0),(r,0))}})r_t
\end{equation}

Using the fact that the functions $f$ and $\frac{\partial f}{\partial x}(x)$ are Lipschitz, for all $0\leq t\leq T(\alpha+b+b_x,s_0)$, it holds that
\begin{equation}
\begin{split}
&\left\Vert \frac{\partial f}{\partial x}|_{\Theta_{(x^a,x^a+\zeta)}}-\frac{\partial f}{\partial x}|_{\Theta_{(x,x+r+\zeta)}} \right\Vert  \leq \\& L'_f \big(\Theta_{(x^a,x^a+\zeta)}-\Theta_{(x,x+r+\zeta)}\big) \leq L'_f(\Vert x_t^a\Vert +\Vert x_t\Vert+\Vert r_t\Vert \\ 
&+\Vert \zeta_t\Vert)\leq L'_f(\alpha+b_x+b_{\zeta})+L'_f\Vert r_t\Vert, 
\end{split}
\end{equation}
where we used our assumption that $\Vert x_t\Vert \leq b_x$ with probability $\delta(T(\alpha+b+b_x,s_0),b_x,b_v)$. Moreover, Lipschitz assumption on $f$ and the boundedness of $\zeta$ also results in 
\begin{equation}
\begin{split}
\left\Vert \frac{\partial f}{\partial x}|_{\Theta_{(x^a,x^a+\zeta)}}-\frac{\partial f}{\partial x}|_{\Theta_{(x,x+r+\zeta)}}\right\Vert  \leq  2L_f 
\end{split}
\end{equation}
Therefore, we get $\Vert \frac{\partial f}{\partial x}|_{\Theta_{(x^a,x^a+\zeta)}}-\frac{\partial f}{\partial x}|_{\Theta_{(x,x+r+\zeta)}}\Vert \leq  \min\{2L_f,L'_f(\alpha+b_x+b_{\zeta})+L'_f\Vert r_t\Vert \}\leq \min\{2L_f,L'_f(\alpha+b_x+b_{\zeta}) \}+L'_f\Vert r_t\Vert$. 
Similarly, we have 
\begin{equation}
\begin{split}
\Vert \frac{\partial f}{\partial x}|&_{\Theta_{(x,x+r+\zeta)}}-\frac{\partial f}{\partial x}|_{\Theta_{(0,r)}}\Vert  \leq  L'_f \big(\Theta_{(x,x+r+\zeta)}-\Theta_{(0,r)}\big) \leq 
\\
&L'_f(\Vert x_t\Vert+\Vert r_t\Vert+\Vert \zeta_t\Vert) \leq  L'_f(b_x+b_{\zeta})+L'_f\Vert r_t\Vert 
\end{split}
\end{equation}
Therefore, $\Vert \sigma_1'\Vert\leq \min\{2L_f,L'_f(\alpha+b_x+b_{\zeta}) \}b_{\zeta}+L'_fb_{\zeta}\Vert r_t\Vert +L'_f(b_x+b_{\zeta})\Vert r_t\Vert+L'_f\Vert r_t\Vert^2$. Now, consider the bounded set of $B_{\phi}$ satisfying $B_{\phi}\subset \mathcal{D}$, where we know $\Vert r_t \Vert^2\leq \phi\Vert r_t \Vert$ for all $r_t\in B_{\phi}$, and we get
\begin{equation}
\begin{split}
\Vert \sigma'_1\Vert\leq& \min\{2L_f,L'_f(\alpha+b_x+b_{\zeta}) \}b_{\zeta}\\&+L'_f(b_x+2b_{\zeta}+\phi)\Vert r_t\Vert
= L_1 \Vert r_t\Vert +L_2 b_{\zeta}.
\end{split}
\end{equation}

In addition, since the function $\Pi'$ is  Lipschitz, for all $x\in \mathcal{D}$ and $r_t\in B_{\phi}$ with probability $\delta(T(\alpha+b+b_x,s_0),b_x,b_v)$ it holds that
\begin{align*}
\left\Vert \frac{\partial \Pi}{\partial e_t}|_{\Theta_{((x,v),(x+r,v))}}-\frac{\partial \Pi}{\partial r_t}|_{\Theta_{((0,0),(0,r))}} \right\Vert &\leq L'_{\Pi}(\Vert x_t+r_t+v_t^s \Vert) \nonumber \\ &\leq L'_{\Pi}(\phi+b_x+b_v),
\end{align*}
and we get $\Vert \sigma_2'\Vert\leq L_{3}\Vert B\Vert \Vert r_t\Vert$ with $L_{3}=L'_{\Pi}(b_x+\phi+b_v)$; this results in 
\begin{equation}
\Vert \sigma'_1+\sigma'_2\Vert\leq L_{2}b_{\zeta}+(L_{1}+L_{3}\Vert B\Vert)\Vert r_t\Vert . 
\end{equation}

Since for $r_t\in B_{\phi}$ we have that $L_{1}+L_{3}\Vert B\Vert < \frac{c_3}{c_4}$ and $L_{2}b_{\zeta}<\frac{c_3-(L_{1}+L_{3}\Vert B\Vert)c_4}{c_4}\sqrt{\frac{c_1}{c_2}}\theta d$, using Lemma~\ref{lemma:khalil} for all $\Vert r_{0}\Vert=\Vert s_{0}\Vert <\min\{\sqrt{\frac{c_1}{c_2}}d,\phi\}$, there exists $t_1> 0$, such that for all $t< t_1$ it holds that $\Vert r_t\Vert \leq \sqrt{\frac{c_2}{c_1}} e^{-\beta t} \Vert s_0\Vert $ with $\beta=\frac{(1-\theta) \big(c_3-(L_{1}+L_{3}\Vert B\Vert)c_4\big)}{2c_2}$,  and  $\Vert r_{t}\Vert\leq b$ with $b=\frac{c_4}{c_3-(L_{1}+L_{3}\Vert B\Vert)c_4}\sqrt{\frac{c_2}{c_1}}\frac{L_{2}b_{\zeta}}{\theta}$ for $t\geq  t_1 $. It should be noted that our assumption $\phi>b$ (or equivalently $B_{b}\subset B_{\phi}$) guarantees that the trajectory of $r_t$ will remain in the set $B_{\phi}$.

Now, we need to show that for $t\geq T(\alpha+b_x+b,s_0)$ we get that $\Vert x_t^a \Vert \geq \alpha$. Since the function $f$ is differentiable, using the Mean-value theorem we obtain
\begin{equation}\label{eq:unstable_dyn}
\begin{split}
s_{t+1}=&f(s_t+\zeta_t+e_t)-f(\zeta_t+e_t)=\\
=&f(s_t)+\frac{\partial{f}}{\partial{x}}|_{(s_t,s_t+e_t+\zeta_t)}(e_t+\zeta_t)-f(e_t+\zeta_t).
\end{split}
\end{equation}
Since $e_{t}=r_t+x_t$, for $0\leq t \leq T(\alpha+b_x+b,s_0)$ with probability $\delta(T(\alpha+b_x+b,s_0),b_x,b_v)$ we have that 
\begin{align*}
\left\Vert \frac{\partial{f}}{\partial{x}}|_{(s_t,s_t+e_t+\zeta_t)}(e_t+\zeta_t)-f(e_t+\zeta_t) \right\Vert & \leq 2L_f\Vert (r_{t}+x_{t}+\zeta_t)\Vert \\
&\leq 2L_f(b_x+b+b_{\zeta}).
\end{align*}

Since we assume that $f\in \mathcal{U}_{\rho}$ with $\rho=2L_f(b_x+b+b_{\zeta})$, there exists $s_0$ such that $s_t$ 
becomes arbitrarily large. Using Definition~\ref{def:unstable_fun}, $T(\alpha+b_x+b,s_0)$ is defined as the first time step that satisfies $\Vert s_t\Vert \geq \alpha+b_x+b$. On the other hand, using $e_t=x_t^a-s_t$ and $e_t=x_t+r_t$, it holds that 
$$\Vert x_t^a\Vert \geq \Vert s_t\Vert - \Vert e_t\Vert \geq \Vert s_t\Vert - \Vert x_t\Vert-\Vert r_t\Vert\geq \alpha.$$ 

We now need to show that the designed $Y_t^a$ for $t\geq 0$ satisfies the stealthiness condition; i.e., we need to show that $KL\big(\mathbf{Q}(Y_{-\infty}^{-1},Y_{0}^a:Y_t^a)||\mathbf{P}(Y_{-\infty}:Y_t)\big)\leq \log(\frac{1}{1-\epsilon^2})$ for some $\epsilon>0$. Since the sequences $Y_{0}^a: Y_{t}^a$, $Y_{0}:Y_{t}$, and $Y_{-\infty}^{-1}$ are generated by $e_{0}:e_{t}$, $x_{0}:x_{t}$, and $x_{-\infty}:x_{-1}$, respectively, using the Data-processing inequality of the KL divergence it holds that
\begin{equation}\label{eq:data_process1}
\begin{split}
&KL\big(\mathbf{Q}({Y}_{-\infty}^{-1},{Y}_{0}^a:{Y}_{t}^a)||\mathbf{P}({Y}_{-\infty}:{Y}_{t})\big)\leq\\
 & KL\big(\mathbf{Q}(x_{-\infty}:x_{-1},e_{0}:e_t)||\mathbf{P}(x_{-\infty}:x_t)\big).
\end{split}
\end{equation}
On the other hand, by defining  ${Z}_{t}=\begin{bmatrix}x_{t}\\y_t^{c,s}\end{bmatrix}$ and ${Z}_{t}^e=\begin{bmatrix}e_{t}\\y_t^{c,s,a}\end{bmatrix}$, and using monotonicity from Lemma~\ref{lemma:mon} it holds that
\begin{equation}\label{eq:data_pro}
\begin{split}
KL\big(\mathbf{Q}(&x_{-\infty}:x_{-1},e_{0}:e_t)||\mathbf{P}(x_{-\infty}:x_t)\big)\leq\\ &KL\big(\mathbf{Q}(Z_{-\infty}:Z_{-1},Z_0^e:Z_t^e)||\mathbf{P}(Z_{-\infty}:Z_{t})\big).
\end{split}
\end{equation}
Then, we apply the chain-rule property of KL-divergence on the right-hand side of 
\eqref{eq:data_pro} to obtain the following
\begin{equation}\label{eq:monotonity}
\begin{split}
&KL\big(\mathbf{Q}(Z_{-\infty}:Z_{-1},Z_0^e:Z_t^e)||\mathbf{P}(Z_{-\infty}:Z_{t})\big)=\\ &KL\big(\mathbf{Q}({Z}_{-\infty}:{Z}_{-1})||\mathbf{P}({Z}_{-\infty}:{Z}_{-1})\big)+\\
&\,\,\,\,\,\,\,KL\big(\mathbf{Q}({Z}_{0}^e:{Z}_{t}^e|{Z}_{-\infty}:{Z}_{-1})||\mathbf{P}({Z}_{0}:{Z}_{t}|{Z}_{-\infty}:{Z}_{-1})\big)\\
&=KL\big(\mathbf{Q}({Z}_{0}^e:{Z}_{t}^e|{Z}_{-\infty}:{Z}_{-1})||\mathbf{P}({Z}_{0}:{Z}_{t}|{Z}_{-\infty}:{Z}_{-1})\big);
\end{split}
\end{equation}
here, we used the fact that the 
KL-divergence of two identical distributions (i.e., $\mathbf{Q}({Z}_{-\infty}:{Z}_{-1})$ and $\mathbf{P}({Z}_{-\infty}:{Z}_{-1})$ since the system is not under attack for $t<0$) is zero. 
Using the chain rule property of the KL divergence we have that 
\begin{equation*}
\begin{split}
K&L\big(\mathbf{Q}({Z}_{0}^e:{Z}_{t}^e|{Z}_{-\infty}:{Z}_{-1})||\mathbf{P}({Z}_{0}:{Z}_{t}|{Z}_{-\infty}:{Z}_{-1})\big)\\
&= KL\big(\mathbf{Q}(e_0|{Z}_{-\infty}:{Z}_{-1})||\mathbf{P}(x_{0}|{Z}_{-\infty}:{Z}_{-1})\big)\\
&+ KL\big(\mathbf{Q}(y_0^{c,s,a}|e_0,{Z}_{-\infty}:{Z}_{-1})||\mathbf{P}(y_{0}^{c,s}|x_0,{Z}_{-\infty}:\mathbf{Z}_{-1})\big)\\
&+...+ KL\big(\mathbf{Q}(e_t|{Z}_{-\infty}:{Z}_{t-1}^e)||\mathbf{P}(x_{t}|{Z}_{-\infty}:{Z}_{t-1})\big)\\
&+ KL\big(\mathbf{Q}(y_t^{c,s,a}|e_t,{Z}_{-\infty}:{Z}^e_{t-1})||\mathbf{P}(y_{t}^{c,s}|x_t,{Z}_{-\infty}:{Z}_{t-1})\big).
\end{split}
\end{equation*}

Given ${Z}_{-\infty}:{Z}_{t-1}$, the distribution of $x_t$ is a Gaussian with some mean $\mu(x_{t-1})$ and covariance $\Sigma_w$ written as $x_t=\mu(x_{t-1})+w_{t-1}$. Similarly using~\eqref{eq:e.r} given  ${Z}_{-\infty}:{Z}_{-1},{Z}_{0}^e:{Z}_{t-1}^e$,  the distribution of $e_t$ is a Gaussian with some mean $\mu(e_{t-1})$ and covariance $\Sigma_w$ written as $e_t=\mu(e_{t-1})+w_{t-1}$. Therefore, we get $r_t=e_t-x_t=\mu(e_{t-1})-\mu(x_{t-1})$. Using Lemma~\ref{lemma:Guassian} and Lemma~\ref{lemma:chain} for all $t\geq 0$ it holds that
\begin{equation}\label{ineq:exp1}
\begin{split}
KL\big(\mathbf{Q}&(e_t|{Z}_{-\infty}:{Z}_{t-1}^e)||\mathbf{P}(x_{t}|{Z}_{-\infty}:{Z}_{t-1})\big)= \\
&=\mathbf{E}_{\mathbf{Q}({Z}_{-\infty}:{Z}_{t-1}^e)} \{r_t^T\Sigma_w^{-1}r_t\}\leq\\ 
&\leq \mathbf{E}_{\mathbf{Q}(e_t|{Z}_{-\infty}:{Z}_{t-1}^e)} \{\lambda_{max}(\Sigma_w^{-1})\Vert r_t\Vert^2\}.
\end{split}
\end{equation}

On the other hand, given $x_t$, the distribution of $y_t^{c,s}$  is a Gaussian with 
mean $C_sx_t$ and covariance $\Sigma_{v^s}$. Similarly, given $e_t$, the distribution of $y_t^{c,s,a}$ is a Gaussian with mean $C_se_t$ and covariance $\Sigma_{v^s}$. Using Lemma~\ref{lemma:Guassian} and Lemma~\ref{lemma:chain} for all $t\geq 0$ 
\begin{equation}\label{ineq:exp2}
\begin{split}
KL\big(\mathbf{Q}&(y_t^{c,s,a}|e_t,{Z}_{-\infty}:{Z}^e_{t-1})||\mathbf{P}(y_{t}^{c,s}|x_t,{Z}_{-\infty}:{Z}_{t-1})\big)=\\
&KL\big(\mathbf{Q}(y_t^{c,s,a}|e_t||\mathbf{P}(y_{t}^{c,s}|x_t)\big)\leq \\
&\mathbf{E}_{\mathbf{Q}(y_t^{c,s,a}|e_t)} \{r_t^TC_s^T\Sigma_{v^s}^{-1}C_sr_t\}
\end{split}
\end{equation}

The arguments inside the expectations in the right hand side of the inequalities~\eqref{ineq:exp1} and~\eqref{ineq:exp2} are upper bounded by $\lambda_{max}(\Sigma_w^{-1})\Vert r_t\Vert^2$ and $\lambda_{max}(C_s^T\Sigma_{v^s}^{-1}C_s)\Vert r_t\Vert^2$, respectively, using the norm property $x^TQx\leq \lambda_{max}(Q)\Vert x\Vert^2$. On the other hand, since $\Vert r_t\Vert$ is bounded from above, using Lemma~\ref{lemma:maximum} we can find an upper bound for~\eqref{ineq:exp1} and~\eqref{ineq:exp2}. 

Specifically, if $T(\alpha+b_x+b,s_0)< t_1 $, then \begin{align*}
\sum_{i=0}^{T(\alpha+b_x+b, s_0)}\hspace{-14pt}\Vert r_i\Vert^2 \hspace{-2pt}\leq&\hspace{-2pt} \min\hspace{-2pt}\left\{T(\alpha+b_x+b, s_0)+1,\sqrt{\frac{c_2}{c_1}}\frac{e^{-\beta}}{1-e^{-\beta}}\right\}\\&\times\Vert s_0\Vert^2
\end{align*}
with probability $\delta(T(\alpha+b_x+b, s_0),b_x,b_v)$. However, if $T(\alpha+b_x+b,s_0)\geq t_1 $ then
\begin{align*}
\sum_{i=0}^{T(\alpha+b_x+b, s_0)}\Vert r_i\Vert^2 &\leq \min\left\{t_1,\sqrt{\frac{c_2}{c_1}}\frac{e^{-\beta}}{1-e^{-\beta}}\right\}\Vert s_0\Vert^2\\
& +(T(\alpha+b_x+b,s_0)+1-t_1)b
\end{align*}
with probability $\delta(T(\alpha+b_x+b, s_0),b_x,b_v)$. Using the inequalities~\eqref{eq:data_process1}--\eqref{ineq:exp2}, we get
\begin{equation*}
\begin{split}
 K&L\big(\mathbf{Q}({Y}_{-\infty}^{-1},Y_{0}^a:Y_{T(\alpha+b_x+b,s_0)}^a)||\mathbf{P}(Y_{-\infty}:Y_{T(\alpha+b_x+b,s_0)})\big)  \\ &\leq \lambda_{max}(C_s^T\Sigma^{-1}_v C_s+\Sigma^{-1}_w)\times \\
 & \max\Big\{\min\big\{T(\alpha+b_x+b, s_0)+1,\sqrt{\frac{c_2}{c_1}}\frac{e^{-\beta}}{1-e^{-\beta}}\big\}\Vert s_0\Vert^2, \\
 & \min\big\{t_1,\sqrt{\frac{c_2}{c_1}}\frac{e^{-\beta}}{1-e^{-\beta}}\big\}\Vert s_0\Vert^2+(T(\alpha+b_x+b,s_0)+1-t_1)b\Big\}\\
 &=b_{\epsilon},
\end{split}
\end{equation*}
which means that the system is ($\epsilon$,$\alpha$)-attackable with probability of $\delta(T(\alpha+b_x+b,s_0),b_x,b_v)$ and $\epsilon=\sqrt{1-e^{-b_{\epsilon}}}$.
\end{proof}

\subsection{Proof of Theorem~\ref{thm_openloop}}
\label{app:t4}

\begin{proof}[Theorem~\ref{thm_openloop}]
We need to show that the sequence of compromised perception (e.g., images) and sensor values $\{z_{0}^a,y_{0}^{s,a}\},\{z_{1}^a,y_{1}^{s,a}\},...$ obtained by Attack Strategy~${\textit{\rom{2}}}$ are ($\epsilon$,$\alpha$)-successful attack. By defining $e_t=x_t^a-s_t$ and $r_t=e_t-x_t$, we get $z^a_t=G(e_t)$, $ y_{t}^{s,a}=C_se_t+v_t^{s,a}$ and 
\begin{equation*}
\begin{split}
e_{t+1}=&f(x_t^a)-f(s_t)+B\Pi(e_t,v_{t}^{s})+w_t^a,\\
r_{t+1}=&f(r_t)+B\Pi(r_t,0)+f(x_t^a)-f(s_t)-f(x_t)\\
&-f(r_t)+B\Pi(e_t,v_{t}^{s})-B\Pi(x_t,v_{t}^{s})-B\Pi(r_t,0)\\
&=h(r_t,0)+\sigma_1+\sigma_2,
\end{split}
\end{equation*}
with $\sigma_1=f(x_t^a)-f(s_t)-f(x_t)-f(r_t)$ and $\sigma_2=B\Pi(e_t,v_{t}^{s})-B\Pi(x_t,v_{t}^{s})-B\Pi(r_t,0)$. Using the Mean value theorem and equality $x_t^a=r_t+s_t+x_t$, 
we obtain
\begin{equation*}
\begin{split}
f(x_t^a)&=f(x_t+s_t+r_t)=f(s_t)+\frac{\partial f}{\partial x}|_{\Theta_{(s,s+x+r)}}(x_t+r_t)\\
f(r_t)&=f(0)+\frac{\partial f}{\partial x}|_{\Theta_{(0,r)}}r_t=\frac{\partial f}{\partial x}|_{\Theta_{(0,r)}}r_t\\
f(x_t)&=f(0)+\frac{\partial f}{\partial x}|_{\Theta_{(0,x)}}x_t=\frac{\partial f}{\partial x}|_{\Theta_{(0,x)}}x_t
\end{split}
\end{equation*}

Therefore, we get
\begin{equation}
\begin{split}
\sigma_1=&(\frac{\partial f}{\partial x}|_{\Theta_{(s,s+x+r)}}-\frac{\partial f}{\partial x}|_{\Theta_{(0,x)}})x_t\\&+(\frac{\partial f}{\partial x}|_{\Theta_{(s,s+x+r)}}-\frac{\partial f}{\partial x}|_{\Theta_{(0,r)}})r_t
\end{split}
\end{equation}
Similarly, using the Mean Value theorem and $e_t=r_t+x_t$, we obtain that
\begin{equation}
\begin{split}
\Pi(e_t,v_{t}^{s})&=\Pi(x_t,v_{t}^{s})+\frac{\partial \Pi}{\partial e_t}|_{\Theta_{(x,x+r)}} r_t\\
\Pi(r_t,0)&=\Pi(0,0)+\frac{\partial \Pi}{\partial r_t}|_{\Theta_{(0,r)}} r_t
\end{split}
\end{equation}
Therefore, we get
\begin{equation}
\sigma_2=B(\frac{\partial \Pi}{\partial e_t}|_{\Theta_{((x,v),(x+r,v))}}-\frac{\partial \Pi}{\partial r_t}|_{\Theta_{((0,0),(r,0))}})r_t.
\end{equation}

Using the fact that the function $\frac{\partial f}{\partial x}(x)$ is Lipschitz, for all $0\leq t\leq T(\alpha+b_x+b,s_0)$ we have  
$$\left\Vert \frac{\partial f}{\partial x}|_{\Theta_{(s,s+x+r)}}-\frac{\partial f}{\partial x}|_{\Theta_{(0,x)}}\right\Vert \leq L'_f(\Vert x_t^a\Vert +\Vert x_t\Vert) \leq L'_f(\alpha+b_x),$$ 
and 
$$\left\Vert \frac{\partial f}{\partial x}|_{\Theta_{(s,s+x+r)}}-\frac{\partial f}{\partial x}|_{\Theta_{(0,r)}}\right\Vert \leq L'_f(\Vert x_t^a\Vert+\Vert r_t\Vert) \leq L'_f(\alpha+\phi).$$
Therefore, $\Vert \sigma_1\Vert\leq L_{2}b_{x}+L_{1}\Vert r_t\Vert$.  

Similarly, 
$\left\Vert \frac{\partial \Pi}{\partial e_t}|_{\Theta_{((x,v),(x+r,v))}}-\frac{\partial \Pi}{\partial r_t}|_{\Theta_{((0,0),(r,0))}} \right\Vert\leq L'_{\Pi}(\Vert x_t+r_t \Vert\leq L'_{\Pi}(\phi+b_x),$ and we get $L_{3}=L'_{\Pi}(b_x+\phi+b_v)$ and $\Vert \sigma_2\Vert\leq L_{3}\Vert r_t\Vert$; this results in $\Vert \sigma_1+\sigma_2\Vert\leq L_{2}b_{x}+(L_{1}+L_{3}\Vert B\Vert)\Vert r_t\Vert$. Since we have $L_{1}+L_{3}\Vert B\Vert < \frac{c_3}{c_4}$ and $L_{2}b_{x}<\frac{c_3-(L_{1}+L_{3}\Vert B\Vert)c_4}{c_4}\sqrt{\frac{c_1}{c_2}}\theta d$, using Lemma~\ref{lemma:khalil} for all $\Vert r_{0}\Vert=\Vert s_{0}\Vert <\min\{\phi, \sqrt{\frac{c_1}{c_2}}d\}$ there exists $t_1> 0$, such that for all $t\geq t_1$ we have $\Vert r_{t}\Vert\leq b$ with $b=\frac{c_4}{c_3-(L_{1}+L_{3}\Vert B\Vert)c_4}\sqrt{\frac{c_2}{c_1}}\frac{L_{2}b_{x}}{\theta}$; also, for all $0\leq t< t_1$, it holds that $\Vert r_t\Vert \leq \sqrt{\frac{c_2}{c_1}} e^{-\beta t} \Vert s_0\Vert $ with $\beta=\frac{(1-\theta) \big(c_3-(L_{1}+L_{3}\Vert B\Vert)c_4\big)}{2c_2}$. 

On the other hand, the dynamics $s_{t+1}=f(s_t)$ with nonzero $s_{0}$ will reach to $\Vert s_t\Vert \geq \alpha+b_x+b$ for some $t\geq T(\alpha+b_x+b,s_0)$ as $f\in \mathcal{U}_0$. Using the reverse triangle inequality, we obtain
\begin{equation*}
\begin{split}
\Vert s_t\Vert-\Vert x_{t}^a \Vert&\leq \Vert x_{t}^a-s_t\Vert=\Vert e_t\Vert=\Vert x_t+r_t \Vert \\
&\leq b_x + b \Rightarrow -b- b_x +b+b_x+\alpha=\alpha  \leq  \Vert x_{t}^a \Vert,
\end{split}
\end{equation*}
with probability $\delta(T(\alpha+b_x+b,s_0),{b_x},b_v)$. 

Now, we need to show that the designed $y_t^a$ satisfies the stealthiness condition. In other words, we need to show that 
$KL\big(\mathbf{Q}(Y_{0}^a:Y_{T(\alpha+b_x+b,s_0)}^a)||\mathbf{P}(Y_{0}:Y_{T(\alpha+b_x+b,s_0)})\big)\leq \log(\frac{1}{1-\epsilon^2})$
for some $\epsilon>0$. Since the sequences $Y_{0}^a,..., Y_{t}^a$ and $Y_{0},..., Y_{t}$ are generated by $e_{0},..., e_{t}$ and $x_{0},..., x_{t}$, respectively, using the Data-processing inequality of KL divergence and following the same procedure as for Theorem~\ref{thm_closeloop}, we obtain
\begin{equation}
\begin{split}
 KL&\big(\mathbf{Q}(Y_{0}^a:Y_{T(\alpha+b_x+b,s_0)}^a)||\mathbf{P}(Y_{0}:Y_{T(\alpha+b_x+b,s_0)})\big) \\
&\leq \sum_{i=0}^{T(\alpha+b_x+b,s_0)}\lambda_{max}(C_s^T\Sigma^{-1}_v C_s+\Sigma^{-1}_w)\Vert r_i \Vert^2.
\end{split}
\end{equation}

Similar argument as in the proof of Theorem~\ref{thm_closeloop} results in
\begin{equation*}
\begin{split}
 &KL\big(\mathbf{Q}(Y_{0}^a:Y_{T(\alpha+b_x+b,s_0)}^a)||\mathbf{P}(Y_{0}:Y_{T(\alpha+b_x+b,s_0)})\big) \leq  \\
& \lambda_{max}(\Sigma_w^{-1})\Vert s_0\Vert^2+\lambda_{max}(C_s^T\Sigma^{-1}_v C_s+\Sigma^{-1}_w)\times\\
 &   \max\Big\{\min\big\{T(\alpha+b_x+b, s_0)+1,\sqrt{\frac{c_2}{c_1}}\frac{e^{-\beta}}{1-e^{-\beta}}\big\}\Vert s_0\Vert^2, \\
 &  \min\big\{t_1,\sqrt{\frac{c_2}{c_1}}\frac{e^{-\beta}}{1-e^{-\beta}}\big\}\Vert s_0\Vert^2+(T(\alpha+b_x+b,s_0)-t_1)b\Big\}\\
 &=b_{\epsilon};
\end{split}
\end{equation*}
this means that the system is ($\epsilon$,$\alpha$)-attackable with probability of $\delta(T(\alpha+b_x+b,s_0),b_x)$ with $\epsilon=\sqrt{1-e^{-b_{\epsilon}}}$.
\end{proof}

\bibliographystyle{IEEEtranMod}


\end{document}